\documentclass[runningheads]{llncs}
\usepackage{graphicx}
\usepackage{verbatim}

\usepackage{amsmath}
\usepackage{amssymb}
\usepackage{stmaryrd} 
\usepackage{wasysym}
\usepackage{mathtools} 
\usepackage{bm}
\usepackage{pifont}
\newcommand{\cmark}{\ding{51}}%
\newcommand{\xmark}{\ding{55}}%
\usepackage{marvosym}
\usepackage[]{algorithm2e}
\usepackage{pgfplots}
\pgfplotsset{compat=1.12}
\usepackage{multicol}

\usepackage{subfig}
\usepackage{wrapfig}

\usepackage{cite}

\usepackage{tikz}
\usetikzlibrary{arrows,automata,positioning}

\usepackage{ltltikz}

\usepackage{hyperref}

\usepackage{listings}
\definecolor{vgreen}{RGB}{104,180,104}
\definecolor{vblue}{RGB}{49,49,255}
\definecolor{vorange}{RGB}{255,143,102}
\lstdefinestyle{verilog-style}
{
	language=Verilog,
	basicstyle=\small\ttfamily,
	keywordstyle=\color{vblue},
	identifierstyle=\color{black},
	commentstyle=\color{vgreen},
	numbers=left,
	numberstyle=\tiny\color{black},
	numbersep=10pt,
	tabsize=8,
	moredelim=*[s][\colorIndex]{[}{]},
	literate=*{:}{:}1
}
\makeatletter
\newcommand*\@lbracket{[}
\newcommand*\@rbracket{]}
\newcommand*\@colon{:}
\newcommand*\colorIndex{%
	\edef\@temp{\the\lst@token}%
	\ifx\@temp\@lbracket \color{black}%
	\else\ifx\@temp\@rbracket \color{black}%
	\else\ifx\@temp\@colon \color{black}%
	\else \color{vorange}%
	\fi\fi\fi
}

\newcommand{\powerset}{\mathcal{P}}
\newcommand{\set}[1]{\{#1\}}

\newcommand{\ldot}{\mathpunct{.}}
\renewcommand{\implies}{\Rightarrow}

\newcommand{\Tau}{\mathcal{T}}

\newcommand{\good}{\mathit{good}}
\newcommand{\bad}{\mathit{bad}}

\newcommand{\pspace}{\textsc{PSpace}}

\newcommand{\ltl}{\text{LTL}}
\newcommand{\hyperltl}{\text{HyperLTL}}
\newcommand{\secltl}{\text{SecLTL}}

\newcommand{\lang}{\mathcal{L}}
\newcommand{\ap}{\text{AP}}
\renewcommand{\models}{\vDash}
\newcommand{\nmodels}{\nvDash}
\newcommand{\traces}{\mathit{TR}}

\newcommand{\pref}{\preceq}
\newcommand{\tracebox}[1]{\framebox[4.5ex][l]{#1}}

\newcommand{\U}{\Until}
\newcommand{\X}{\Next}
\newcommand{\G}{\Globally}

\newcommand{\W}{\WUntil}
\newcommand{\true}{\mathit{true}}

\newcommand{\pathassign}{\Pi}
\newcommand{\pathvars}{\mathcal{V}}
\newcommand{\pathassignfin}{\Pi_\mathit{fin}}

\newcommand{\monitor}{\mathcal{M}}

\newcommand{\ninfluences}{\not\leadsto}
\title{Monitoring Hyperproperties\thanks{This work was partially supported by the European Research Council (ERC) Grant OSARES (No. 683300) and as part of the Collaborative Research Center “Methods and Tools for Understanding and Controlling Privacy” (SFB 1223) by the German Research Foundation (DFG).}
}
\titlerunning{Monitoring Hyperproperties}
\author{Bernd Finkbeiner, Christopher Hahn, Marvin Stenger and Leander Tentrup}
\authorrunning{Finkbeiner, Hahn, Stenger, Tentrup}
\institute{%
	Reactive Systems Group\\
	Saarland University\\
	\email{lastname@react.uni-saarland.de}
}

\begin{document}
	
\maketitle              
Hyperproperties, such as non-interference and observational determinism, relate multiple system executions to each other. They are not expressible in standard temporal logics, like LTL, CTL, and CTL*, and thus cannot be monitored with standard runtime verification techniques. $\hyperltl$ extends linear-time temporal logic (LTL) with explicit quantification over traces in order to express Hyperproperties. We investigate the runtime verification problem of $\hyperltl$ formulas for three different input models: (1) The parallel model, where a fixed number of system executions is processed in parallel. (2) The unbounded sequential model, where system executions are processed sequentially, one execution at a time. In this model, the number of incoming executions may grow forever. (3) The bounded sequential model where the traces are processed sequentially and the number of incoming executions is bounded. We show that deciding monitorability of $\hyperltl$ formulas is $\pspace$-complete for input models (1) and (3). Deciding monitorability is $\pspace$-complete for alternation-free $\hyperltl$ formulas in input model (2). For every input model, we provide practical monitoring algorithms. We also present various optimization techniques. By recognizing properties of specifications such as reflexivity, symmetry, and transitivity, we reduce the number of comparisons between traces. For the sequential models, we present a technique that minimized the number of traces that need to be stored. Finally, we provide an optimization that succinctly represents the stored traces by sharing common prefixes. We evaluate our optimizations, showing that this leads to much more scalable monitoring, in particular, significantly lower memory consumption.


\section{Introduction}
\label{intro}
\sloppy
\emph{Hyperproperties}~\cite{journals/jcs/ClarksonS10} generalize trace properties in that they not only check the correctness of individual traces, but can also relate multiple computation traces to each other. This is needed, for example, to express information flow security policies like the requirement 
that the system behavior appears to be deterministic, i.e., independent of certain secrets, to an external observer. Monitoring hyperproperties is difficult, because it is no longer possible to analyze traces in isolation: a violation of a hyperproperty in general involves a set of traces, not just a single trace.

We present monitoring algorithms for hyperproperties given in the temporal
logic HyperLTL~\cite{conf/post/ClarksonFKMRS14}, which extends linear-time
temporal logic (LTL) with trace variables and trace quantifiers in order
to refer to multiple traces at a time. For example, the HyperLTL formula
$$\forall \pi. \exists \pi'.~ \G dummyInput_{\pi'} \land ~ lowOut_\pi = lowOut_{\pi'}$$
expresses \emph{noninference} by 
stating that for all traces $\pi$, there exists a trace $\pi'$, such that the observable outputs are the same on both traces even when the high security input of $\pi'$ being replaced by a dummy input. For example, in a messaging app, we might replace the address book, which we want to keep secret, with an \emph{empty} address book.

A first, and absolutely fundamental, question to be answered is in what form the input, which now consists
of more than one execution trace, should be presented to the monitor. Should
the traces be presented all at once or one at a time? Is the number of traces
known in advance? Obviously, the choice of the input representation has significant impact  both on the principal monitoriability of a hyperproperty and on the actual monitoring algorithm.

We study three basic input models for monitoring hyperproperties. (1) The \emph{parallel} model, where a \emph{fixed} number of system executions is processed in parallel.
(2) The \emph{unbounded sequential} model, where system executions are processed sequentially, one execution at a time. In this model, the number of incoming executions is a-priori unbounded and may in fact grow forever. (3) The \emph{bounded sequential model} where the traces are processed sequentially and the number of incoming executions is \emph{bounded}.

\paragraph{Parallel model.} The assumption that the number of incoming traces is fixed before the actual monitoring process starts, results in the easiest and most efficient monitoring algorithms. We distinguish \emph{online monitoring}, where the traces become available one position at a time from left to right, from \emph{offline monitoring} where the positions of the traces can be accessed in any order. In particular, offline algorithms can traverse the traces in backwards direction, which is more efficient. Figure~\ref{fig:fixed} illustrates the two types of algorithms. 
\begin{figure}[h]
	\begin{minipage}{0.5\textwidth}
		\centering
		\includegraphics[scale=0.3]{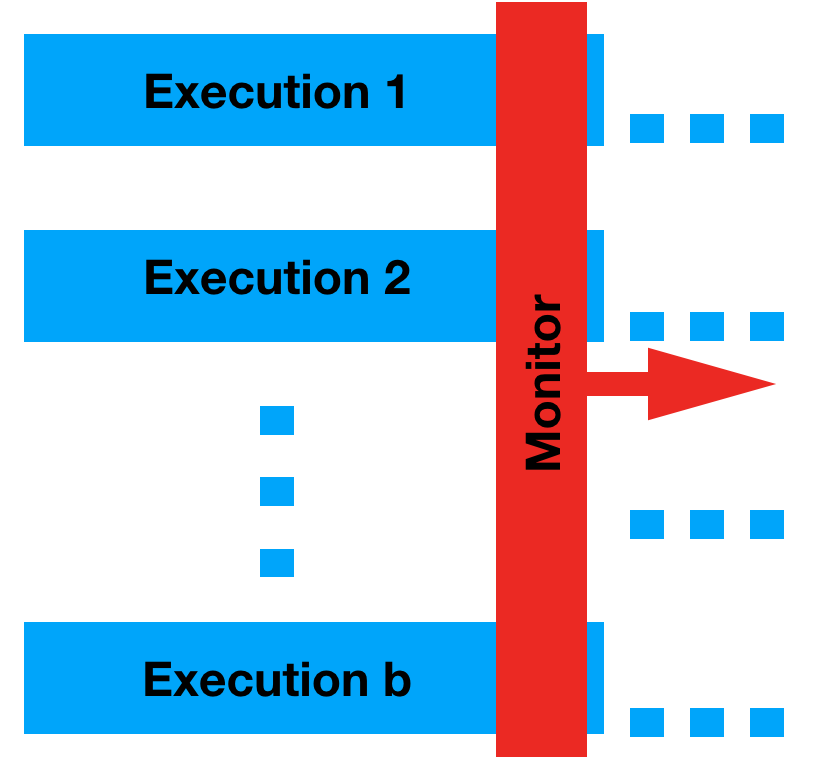}
		
	\end{minipage}\begin{minipage}{0.5\textwidth}
		
		\centering
		\includegraphics[scale=0.3]{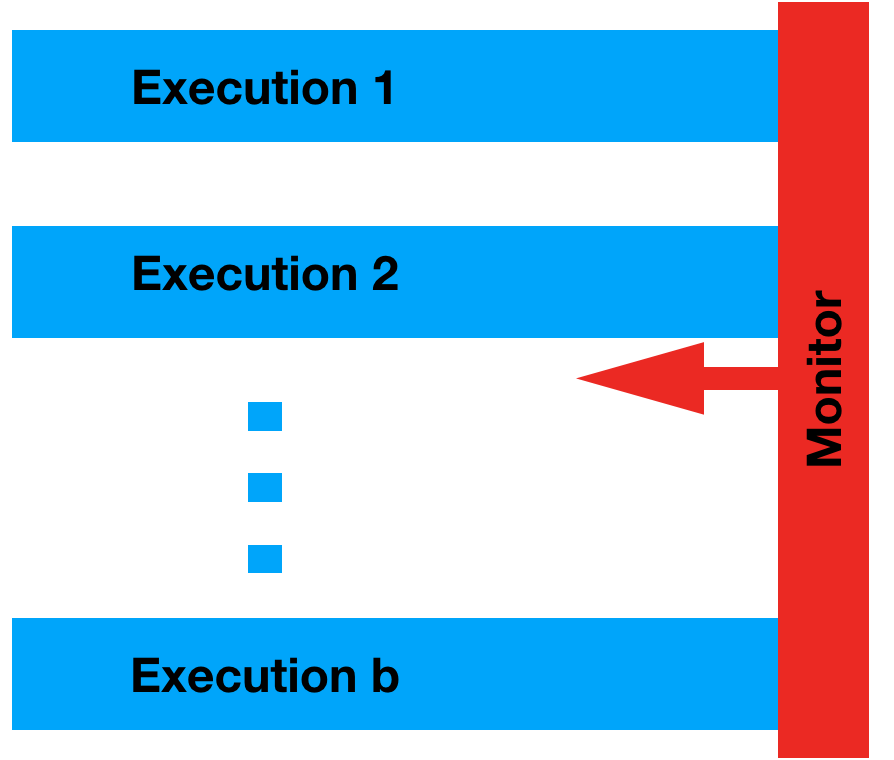}
		
	\end{minipage}
	\caption{Monitor approaches for the fixed size model: online in a forward fashion (left) and offline in a backwards fashion (right).}
	\label{fig:fixed}
\end{figure}
The parallel model is known from techniques like secure-multi-execution~\cite{conf/sp/DevrieseP10}, where several system executions are generated by providing different high-security inputs.
We present an online and an offline monitoring algorithm for hyperproperties expressed in $\hyperltl$. The online algorithm is based on standard techniques for building monitoring automata from LTL formulas. Such a monitor automaton is then instantiated for multiple traces as specified by the HyperLTL formula. 
The offline algorithm is based on constructing an alternating automaton and then proceeding through the automaton in a bottom-up fashion, similar to the classic construction for LTL~\cite{DBLP:journals/fmsd/FinkbeinerS04}.

\paragraph{Unbounded and bounded sequential model.}  The sequential models are useful when multiple sessions of a system under observation have to be monitored one after the other in an online fashion. The disadvantage of the unbounded sequential model is that many interesting hyperproperties, in particular most hyperproperties with quantifier alternations, are not monitorable in this model. It is therefore often useful to define a stop condition in the form of a bound on the number of traces that need to be handled during the monitoring process.
\begin{figure}[h]
	\begin{minipage}{0.5\textwidth}
		
		\centering
		\includegraphics[scale=0.3]{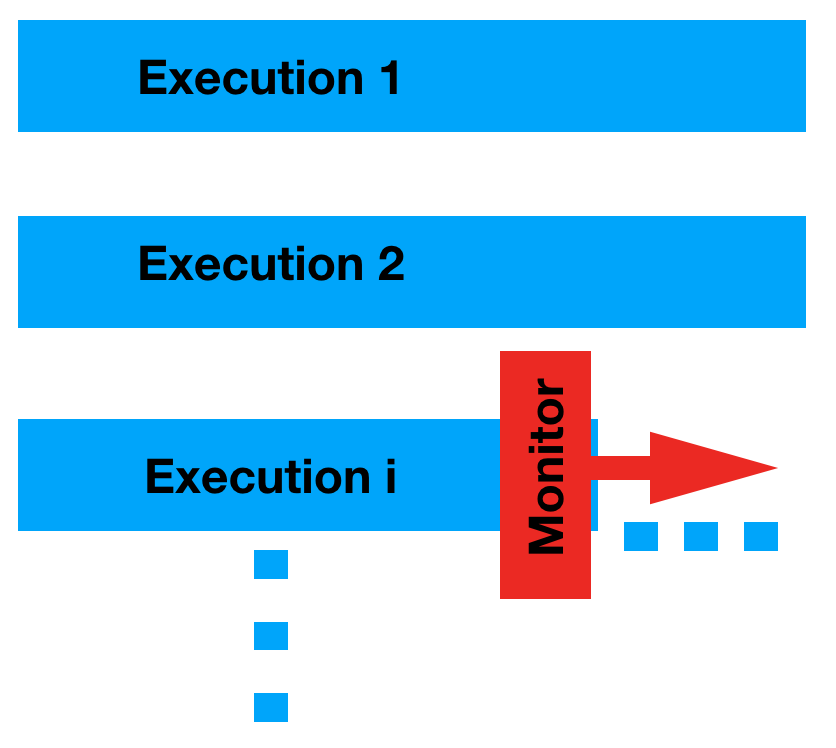}
		
	\end{minipage}\begin{minipage}{0.5\textwidth}
		
		\centering
		\includegraphics[scale=0.3]{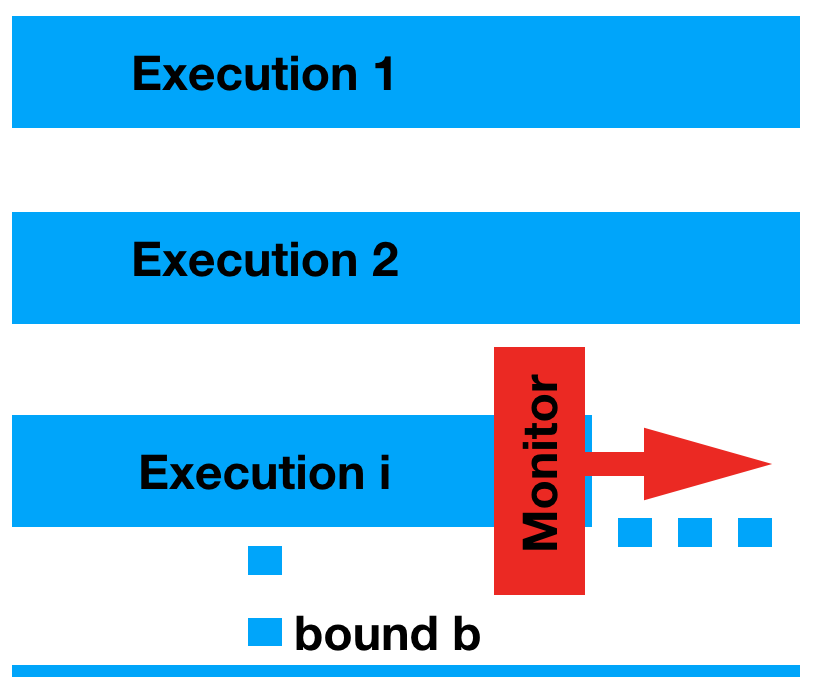}
		
	\end{minipage}
	\caption{Monitor approaches for the sequential models: an unbounded number of traces (left) and bounded number of traces (right) are processed sequentially.}
	\label{fig:sequential}
\end{figure}
Figure~\ref{fig:sequential} sketches the monitoring algorithm for the unbounded and bounded cases. 

A naive monitoring approach for the sequential models would be to simply store all traces seen so far. However, this would create two problems: a memory problem, because the needed memory grows with the number of traces observed by the monitor, and a runtime problem, because one needs to relate every newly observed trace against the growing set of stored traces.

There are hyperproperties where this effect cannot be avoided. An example is the hyperproperty with two atomic propositions $p$ and $q$, where any pair of traces that agree on their $p$ labeling must also agree on their $q$ labeling. Clearly, for every $p$ labeling seen so far, we must also store the corresponding $q$ labeling. 
In practice, however, it is often possible to greatly simplify the monitoring. Consider, for example, the hyperproperty that states that all traces have the same $q$ labeling (independently of the $p$ labeling). In HyperLTL, this property is specified as the formula $$\forall \pi\ldot \forall \pi'\ldot \G (q_\pi {\leftrightarrow} q_{\pi'}).$$ The naive approach would store all traces seen so far, and thus require $O(n)$ memory after $n$ traces. A new trace would be compared against every stored trace twice, once as $\pi$ and once as $\pi'$, resulting in a $O(2n)$ running time for each new trace.
Obviously, however, in this example it is sufficient to store the first trace, and compare all further incoming traces against this reference. The required memory is thus, in fact, constant in the number of traces. A further observation is that the specification is symmetric in $\pi$ and $\pi'$. Hence, a single comparison suffices.

In this article, we present a monitoring approach for hyperproperties in the unbounded model that reduces the set of traces that new traces must be compared against to a minimal subset.
Our approach comes with a strong correctness guarantee: 
our monitor produces the same verdict as a naive monitor that would store all traces and, additionally, we keep a sufficient set of traces to always provide an actually observed witness for the monitoring verdict. Our monitoring thus delivers a result that is equally informative as the naive solution, but is computed faster and with less memory.

We introduce two analysis techniques and an optimized succinct representation of the trace set to be monitored: \emph{The trace analysis} reduces the stored set of traces to a minimum, thus minimizing the required memory. \emph{The specification analysis}, which is applicable in the parallel model as well, identifies symmetry, transitivity, and reflexivity in the specification, in order to reduce the algorithmic workload that needs to be carried out on the stored traces.
\emph{The tries representation} exploits the often prefix-equal traces in the incoming trace set, by storing them in a tree-like data structure called tries.



\paragraph{Trace Analysis.}
As an example for a system where confidentiality and information flow is of outstanding importance for the intended operation, we consider a conference management system.
There are a number of confidentiality properties that such a system should satisfy, like \emph{``The final decision of the program committee remains secret until the notification''} and \emph{``All intermediate decisions of the program committee are never revealed to the author''}.
We want to focus on important hyperproperties of interest beyond confidentiality, like the property that no paper submission is lost or delayed.
Informally, one formulation of this property is \emph{``A paper submission is immediately visible for every program committee member''}.
More formally, this property relates pairs of traces, one belonging to an author and one belonging to a program committee member.
We assume this separation is indicated by a proposition $pc$ that is either disabled or enabled in the first component of those traces.
Further propositions in our example are the proposition $s$, denoting that a paper has been submitted, and $v$ denoting that the paper is visible.

Given a set of traces $T$, we can verify that the property holds by checking every pair of traces $(t,t') \in T \times T$ with $pc \notin t[0]$ and $pc \in t'[0]$ that $s \in t[i]$ implies $v \in t'[i+1]$ for every $i \geq 0$.
When $T$ satisfies the property, $T \cup \set{t^*}$, where $t^*$ is a new trace, amounts to checking new pairs $(t^*,t)$ and $(t,t^*)$ for $t \in T$.
This, however, leads to an increasing size of $T$ and thereby to an increased number of checks: the monitoring problem becomes inevitable costlier over time.
To circumvent this, we present a method that keeps the set of traces \emph{minimal} with respect to the underlying property.
When monitoring hyperproperties, traces may pose \emph{requirements} on future traces.
The core idea of our approach is to characterize traces that pose strictly stronger requirements on future traces than others.
In this case, the traces with the weaker requirements can be safely discarded.
As an example, consider the following set of traces
\begin{align}
&\tracebox{\set{s}}\tracebox{\hspace{0.5ex}\{\}} \tracebox{\hspace{0.5ex}\{\}} \tracebox{\hspace{0.5ex}\{\}} \tracebox{\hspace{0.5ex}\{\}} &&\textit{an author immediately submits a paper} \label{intro-one}\\
&\tracebox{\hspace{0.5ex}\{\}}\tracebox{\set{s}} \tracebox{\hspace{0.5ex}\{\}} \tracebox{\hspace{0.5ex}\{\}} \tracebox{\hspace{0.5ex}\{\}} &&\textit{an author submits a paper after one time unit}\label{intro-two}\\
&\tracebox{\hspace{0.5ex}\{\}}\tracebox{\set{s}} \tracebox{\set{s}} \tracebox{\hspace{0.5ex}\{\}} \tracebox{\hspace{0.5ex}\{\}} &&\textit{an author submits two papers}\label{intro-three}
\end{align}
A satisfying PC trace would be $\set{pc}\set{v}\set{v}\set{v}\emptyset$ as there are author traces with paper submissions at time step 0, 1, and 2.
For checking our property, one can safely discard trace~\ref{intro-two} as it poses no more requirements than trace~\ref{intro-three}.
We say that trace~\ref{intro-three} dominates trace~\ref{intro-two}.
We show that, given a property in the temporal logic $\hyperltl$, we can automatically reduce trace sets to be minimal with respect to this dominance.
On relevant and more complex information flow properties, this reduces the memory consumption dramatically.

\paragraph{Specification Analysis.} 
For expressing hyperproperties, we use the recently introduced temporal logic $\hyperltl$~\cite{conf/post/ClarksonFKMRS14}, which extends linear-time temporal logic (LTL)~\cite{conf/focs/Pnueli77} with explicit trace quantification. We construct a monitor template, containing trace variables, from the $\hyperltl$ formula. We initialize this monitor with explicit traces resulting in a family of monitors checking the relation, defined by the hyperproperty, between the traces.
Our specification analysis technique allows us to reduce the number of monitors in order to detect violation or satisfaction of a given $\hyperltl$ formula.
We use the decision procedure for the satisfiability problem of $\hyperltl$~\cite{conf/concur/FinkbeinerH16} to check whether or not a universally quantified $\hyperltl$ formula is symmetric, transitive, or reflexive.
If a hyperproperty  is \emph{symmetric}, then we can omit every symmetric monitor, thus, performing only half of the language membership tests.
A canonical example for a symmetric $\hyperltl$ formula is $$\forall \pi. \forall \pi'.\; (O_\pi = O_{\pi'}) \W (I_\pi \neq I_{\pi'}),$$ a variant of observational determinism~\cite{journals/jcs/McLean92,conf/sp/Roscoe95,conf/csfw/ZdancewicM03}.
Symmetry is particular interesting, since many information flow policies have this property.
If a hyperproperty is \emph{transitive}, then we can omit every, except for one, monitor, since we can check every incoming trace against any reference trace.
One example for a transitive $\hyperltl$ formula is equality $$\forall \pi. \forall \pi' \ldot \G (a_\pi \leftrightarrow a_{\pi'}).$$
If a hyperproperty is \emph{reflexive}, then we can omit the monitor where every trace variable is initialized with the same trace. For example, both hyperproperties above are reflexive.

\paragraph{Trie representation.}
With the trace analysis, described above, we eliminate traces that are dominated by new incoming traces. The trie representation is an optimization that exploits prefix-equality of traces and therefore succinctly represents traces that are \emph{similar}, although not necessarily equal or dominated. A \emph{trie} is a tree-like data structure that can represent a set of traces as a tree, where traces with equal prefixes collapse to the same path up to the point where the traces differ. Exactly at this position, the trie branches.
For example, consider the finite traces $\tau=\{\mathit{on}\}^{n+1}$ and $\tau'= \{\mathit{on}\}^{n}\{\mathit{off}\}$ and the hyperproperty $\forall \pi. \forall \pi'. \G (\mathit{on}_\pi \oplus \mathit{off}_{\pi'})$. Since none of the traces dominate each other, a monitor would have to spawn four monitor instances $(\tau,\tau'),(\tau',\tau),(\tau,\tau),$ and $(\tau',\tau')$. By using the trie representation, only one monitor instance suffices up to position $n+1$.

\paragraph{Structure of this Article.}

The remainder of this article is structured as follows. Section~\ref{prelim} introduces the syntax and semantics of $\hyperltl$ and the notion of monitorable $\hyperltl$ formula in all three input models.
We furthermore present algorithms for checking whether a $\hyperltl$ formula is monitorable or not.
In Section~\ref{sec:fixed_size_monitoring}, we give a finite trace semantics for $\hyperltl$. For the parallel input model, we present an offline and online monitoring algorithm for arbitrarily $\hyperltl$ formulas.
In Section~\ref{sec:sequential_monitoring}, we present online algorithms for (universal) $\hyperltl$ formulas in the (unbounded) sequential model.
We then tackle the above mentioned memory explosion in by formally introducing the trace analysis, hyperproperty analysis and the tries data structure sketched above. 
We report on our implementation RVHyper v2 and experimental results in Section~\ref{sec:experimentalresults}, before concluding in Section~\ref{sec:conclusion}.

This is a revised and extended version of a paper that appeared at RV~2017~\cite{conf/rv/FinkbeinerHST17}.

\paragraph{Related Work.}

The temporal logic $\hyperltl$ was introduced to model check security properties of reactive systems~\cite{conf/post/ClarksonFKMRS14,conf/cav/FinkbeinerRS15}.
For one of its predecessors, $\secltl$~\cite{conf/vmcai/DimitrovaFKRS12}, there has been a proposal for a white box monitoring approach~\cite{conf/isola/DimitrovaFR12} based on alternating automata.
The problem of monitoring $\hyperltl$ has been considered before~\cite{conf/csfw/AgrawalB16,conf/tacas/BrettSB17}.
Agrawal and Bonakdarpour~\cite{conf/csfw/AgrawalB16} gave a syntactic characterization of monitorable $\hyperltl$ formulas and a monitoring algorithm based on Petri nets.
In subsequent work, a constraint based approach has been proposed~\cite{conf/tacas/BrettSB17}.
Like our monitoring algorithm, they do not have access to the implementation (black box), but in contrast to our work, they do not provide witnessing traces for a monitor verdict.
For certain information flow policies, like non-interference and some extensions, dynamic enforcement mechanisms have been proposed.
Techniques for the enforcement of information flow policies include tracking dependencies at the hardware level~\cite{conf/asplos/SuhLZD04}, language-based monitors~\cite{journals/jsac/SabelfeldM03,conf/csfw/AskarovS09,conf/pldi/AustinF10,conf/csfw/VanhoefGDPR14,conf/post/BichhawatRGH14}, and abstraction-based dependency tracking~\cite{conf/asian/GuernicBJS06,conf/essos/KovacsS12,conf/csfw/ChudnovKN14}.
Secure multi-execution~\cite{conf/sp/DevrieseP10} is a technique that can enforce non-interference by executing a program multiple times in different security levels.
To enforce non-interference, the inputs are replaced by default values whenever a program tries to read from a higher security level.


\section{Runtime Verification of $\hyperltl$}
\label{prelim}
As Hyperproperties relate multiple executions to each other, a monitor for hyperproperties has to consider sets of traces instead of solely processing a single execution in isolation.
In this section, we elaborate on the runtime verification problem of $\hyperltl$.
In the first subsection, we present $\hyperltl$, which is a temporal logic for expressing hyperproperties.
In the second subsection, we define the notion of monitorable $\hyperltl$ specifications for three different input models: the unbounded input model, the bounded model, and the parallel model, which is a special case of the latter.

We begin by defining some notation.
Let $\ap$ be a finite set of atomic propositions and let $\Sigma = 2^\ap$ be the corresponding finite \emph{alphabet}.
A finite (infinite) trace is a finite (infinite) sequence over $\Sigma$.
We denote the concatenation of a finite trace $u \in \Sigma^*$ and a finite or infinite trace $v \in \Sigma^* \cup \Sigma^\omega$ by $uv$ and write $u \pref v$ if $u$ is a prefix of $v$.
Further, we lift the prefix operator to sets of traces, i.e., $U \pref V \coloneqq \forall u \in U \ldot \exists v \in V \ldot u \pref v$ for $U \subseteq \Sigma^*$ and $V \subseteq \Sigma^* \cup \Sigma^\omega$.
We denote the powerset of a set $A$ by $\powerset(A)$ and define $\powerset^*(A)$ to be the set of all finite subsets of $A$.

\subsection{HyperLTL}

$\hyperltl$~\cite{conf/post/ClarksonFKMRS14} is a temporal logic for specifying hyperproperties.
It extends $\ltl$~\cite{conf/focs/Pnueli77} by quantification over trace variables $\pi$ and a method to link atomic propositions to specific traces.
The set of trace variables is $\pathvars$.
Formulas in $\hyperltl$ are given by the grammar
\begin{align*}
\varphi &{}\Coloneqq \forall\pi\ldot\varphi \mid \exists\pi\ldot\varphi \mid \psi \enspace, \text{ and}\\
\psi &{}\Coloneqq a_\pi \mid \neg\psi \mid \psi\lor\psi \mid \X\psi \mid \psi\U\psi \enspace,
\end{align*}
where $a \in \ap$ and $\pi \in \pathvars$.
We call a $\hyperltl$ formula an $\ltl$ formula if it is quantifier free.
The semantics is given by the satisfaction relation $\models_P$ over a set of traces $T \subseteq \Sigma^\omega$.
We define an assignment $\pathassign : \pathvars \to \Sigma^\omega$ that maps trace variables to traces. $\pathassign[i,\infty]$ denotes the trace assignment that is equal to $\pathassign(\pi)[i,\infty]$ for all $\pi$.
\begin{equation*}
\begin{array}{ll}
\pathassign \models_T a_\pi         \qquad \qquad & \text{if } a \in \pathassign(\pi)[0] \\
\pathassign \models_T \neg \varphi              & \text{if } \pathassign \nmodels_T \varphi \\
\pathassign \models_T \varphi \lor \psi         & \text{if } \pathassign \models_T \varphi \text{ or } \pathassign \models_T \psi \\
\pathassign \models_T \X \varphi                & \text{if } \pathassign[1,\infty] \models_T \varphi \\
\pathassign \models_T \varphi\U\psi             & \text{if } \exists i \geq 0 \ldot \pathassign[i,\infty] \models_T \psi \land \forall 0 \leq j < i \ldot \pathassign[j,\infty] \models_T \varphi \\
\pathassign \models_T \exists \pi \ldot \varphi & \text{if there is some } t \in T \text{ such that } \pathassign[\pi \mapsto t] \models_T \varphi
\end{array}
\end{equation*}
We write $T \models \varphi$ for $\set{} \models_T \varphi$ where $\set{}$ denotes the empty assignment.
The language of a $\hyperltl$ formula $\varphi$, denoted by $\lang(\varphi)$, is the set $\set{ T \subseteq \Sigma^\omega \mid T \models \varphi }$.
Let $\varphi$ be a $\hyperltl$ formula with trace variables $\pathvars = \set{\pi_1,\dots,\pi_k}$ over alphabet $\Sigma$.
We define $\Sigma_\pathvars$ to be the alphabet where $p_\pi$ is interpreted as an atomic proposition for every $p \in \ap$ and $\pi \in \pathvars$.
We denote by $\models_\ltl$ the $\ltl$ satisfaction relation over $\Sigma_\pathvars$.
We define the $\pi$-projection, denoted by $\#_\pi(s)$, for a given $s \subseteq \Sigma_\pathvars$ and $\pi \in \pathvars$, as the set of all $p_\pi \in s$.

\begin{lemma} \label{thm:hyperltl-body-to-ltl}
	Let $\psi$ be an $\ltl$ formula over trace variables $\pathvars$.
	There is a trace assignment $A$ such that $A \models_\emptyset \psi$ if, and only if, $\psi$ is satisfiable under $\ltl$ semantics over atomic propositions $\Sigma_\pathvars$.
	The models can be translated effectively.
\end{lemma}
\begin{proof}
	Assume that there is a trace assignment $A$ over trace variables $\pathvars$ such that $A \models_\emptyset \psi$.
	We define $w \subseteq \Sigma_\pathvars^\omega$ such that $x_\pi \in w[i]$ if, and only if, $x \in A(\pi)[i]$ for all $i \geq 0$, $x \in \ap$, and $\pi \in \pathvars$.
	An induction over $\psi$ shows that $w \models_\ltl \psi$.
	
	Assume $\psi$ is satisfiable for $\models_\ltl$, i.e., there exists a $w \subseteq \Sigma^\omega_\pathvars$, such that $w \models_\ltl \psi$.
	We construct an assignment $A$ in the following manner: Let $\pi \in \pathvars$ be arbitrary.
	We map $\pi$ to the trace $t$ obtained by projecting the corresponding $p_\pi \in \Sigma_\pathvars$, i.e., $\forall i \geq 0 \ldot t[i] = \#_\pi(w[i])$.
	\qed
\end{proof}

\subsection{Monitorability} \label{sec:monitorability}

In general, we distinguish between three different input models.

In the first model, the size of the trace set is not known in advance. Traces are processed sequentially, i.e., given a set of traces $T$ and a fresh trace $t$, the monitoring algorithm has to decide whether $T \cup t$ satisfies or violates a given $\hyperltl$ formula before processing another fresh traces $t'$.

In the second model, we consider the special case of the sequential model, where a maximal bound on the trace set is known.
Figure~\ref{fig:sequential} sketches how the monitor processes the traces sequentially for the unbounded and bounded case.

In the third input model, which is a special case of the bounded model, we assume that the set of execution traces is of fixed size and, furthermore, that the traces are given at the same time. This models, for example, secure multi execution, where multiple instances of a system are run in parallel.
As shown in~\cite{DBLP:journals/fmsd/FinkbeinerS04} offline monitoring of future-time temporal formulas can be monitored efficiently by processing a trace in a backwards fashion.
Figure~\ref{fig:fixed} sketches how a monitor will process the traces in an online and offline fashion.


In the remainder of this section, we present three notions of monitorability for hyperproperties for every input model respectively.
In the first subsection, we will consider monitorability of $\hyperltl$ formulas in the unbounded sequential model, before considering the bounded sequential model and the parallel model.
We show that deciding whether alternation-free $\hyperltl$ formulas are monitorable in the unbounded fragment is $\pspace$-complete, i.e., no harder than the corresponding problem for $\ltl$.
Furthermore, we show that deciding whether an arbitrary $\hyperltl$ formula is monitorable in the bounded input model is $\pspace$-complete as well.
Those results extend earlier characterizations based on restricted syntactic fragments of $\hyperltl$~\cite{conf/csfw/AgrawalB16}.

For trace languages, monitorability is the property whether language containment can be decided by finite prefixes~\cite{conf/fm/PnueliZ06}.
Given a trace language $L \subseteq \Sigma^\omega$, the set of \emph{good} and \emph{bad} prefixes is $\good(L) \coloneqq \set{u \in \Sigma^* \mid \forall v \in \Sigma^\omega \ldot uv \in L}$ and $\bad(L) \coloneqq \set{u \in \Sigma^* \mid \forall v \in \Sigma^\omega \ldot uv \notin L}$, respectively.
$L$ is \emph{monitorable} if $\forall u \in \Sigma^* \ldot \exists v \in \Sigma^* \ldot uv \in \good(L) \lor uv \in \bad(L)$.
The decision problem, i.e., given an $\ltl$ formula $\varphi$, decide whether $\varphi$ is monitorable, is $\pspace$-complete~\cite{journals/corr/abs-1006-3638}.

A \emph{hyperproperty} $H$ is a set of trace properties, i.e., $H \subseteq \powerset(\Sigma^\omega)$.
Given $H \subseteq \powerset(\Sigma^\omega)$.
The set of \emph{good} and \emph{bad prefix traces} is $\good(H) \coloneqq \set{U \in \powerset^*(\Sigma^*) \mid \forall V \in \powerset(\Sigma^\omega) \ldot U \pref V \implies V \in H}$ and $\bad(H) \coloneqq \set{U \in \powerset^*(\Sigma^*) \mid \forall V \in \powerset(\Sigma^\omega) \ldot U \pref V \implies V \notin H}$, respectively.

\paragraph{Unbounded Sequential Model.} A hyperproperty $H$ is \emph{monitorable} in the unbounded input model if
\begin{equation*}
\forall U \in \powerset^*(\Sigma^*) \ldot \exists V \in \powerset^*(\Sigma^*) \ldot U \pref V \implies V \in \good(H) \lor V \in \bad(H) \enspace.
\end{equation*}

With this definition, hardly any alternating $\hyperltl$ formula is monitorable as their satisfaction cannot be characterized by a finite trace set, even for safety properties.
Consider, for example, the formula $\varphi = \forall \pi \ldot \exists \pi' \ldot \G (a_\pi \rightarrow b_{\pi'})$.
Assume a finite set of traces $T$ does not violate the formula.
Then, one can construct a new trace $t$ where $a \in t[i]$ and $b \notin t[i]$ for some position $i$, and for all traces $t' \in T$ it holds that $b \notin t'[i]$.
Thus, the new trace set violates $\varphi$.
Likewise, if there is a finite set of traces that violates $\varphi$, a sufficiently long trace containing only $b$'s stops the violation.

However, we present a method to decide whether a $\hyperltl$ formula in the highly expressive fragment of quantifier alternation-free formulas is monitorable.


\begin{lemma}\label{thm:forall-good-empty}
	Given a $\hyperltl$ formula $\varphi = \forall \pi_1 \dots \forall \pi_k \ldot \psi$, where $\psi$ is an $\ltl$ formula.
	It holds that $\good(\lang(\varphi)) = \emptyset$ unless $\psi \equiv \true$.
\end{lemma}
\begin{proof}
	If $\psi \equiv \true$ then $\lang(\varphi) = \powerset(\Sigma^\omega)$ and $\good(\lang(\varphi)) = \powerset^*(\Sigma^*)$.
	Assume for contradiction that $\psi \not\equiv \true$ and $\good(\lang(\varphi)) \neq \emptyset$, i.e., there is a finite set $U \subseteq \Sigma^*$ that is a good prefix set of $\varphi$.
	Since $\psi \neq \true$, there is at least one infinite trace $\sigma$ with $\sigma \nmodels \psi$.
	We translate this trace to a set of infinite traces $W$ where $W \nmodels \varphi$ using Lemma~\ref{thm:hyperltl-body-to-ltl}.
	Further, for all $V \in \powerset(\Sigma^*)$ with $U \pref V$, it holds that $W \subseteq V$, hence, $V \notin \lang(\varphi)$ violating the assumption that $U \in \good(\lang(\varphi))$.
	\qed
\end{proof}

\begin{theorem} \label{thm:monitorability-forall}
	Given a $\hyperltl$ formula $\varphi = \forall \pi_1 \dots \forall \pi_k \ldot \psi$, where $\psi \not\equiv \true$ is an $\ltl$ formula.
	$\varphi$ is monitorable if, and only if, $\forall u \in \Sigma_\pathvars^* \ldot \exists v \in \Sigma_\pathvars^* \ldot uv \in \bad(\lang(\psi))$.
\end{theorem}
\begin{proof}
	Assume $\forall u \in \Sigma_\pathvars^* \ldot \exists v \in \Sigma_\pathvars^* \ldot uv \in \bad(\lang(\psi))$ holds.
	Given an arbitrary prefix $U \in \powerset^*(\Sigma^*)$.
	Pick an arbitrary mapping from $U$ to $\Sigma^*_\pathvars$ and call it $u'$.
	By assumption, there is a $v' \in \Sigma^*_\pathvars$ such that $u'v' \in \bad(\lang(\psi))$.
	We use this $v'$ to extend the corresponding traces in $U$ resulting in $V \in \powerset^*(\Sigma^*)$.
	It follows that for all $W \in \powerset(\Sigma^\omega)$ with $V \pref W$, $W \nmodels \varphi$, hence, $V \in \bad(\lang(\varphi))$.

	Assume $\varphi$ is monitorable, thus, $\forall U \in \powerset^*(\Sigma^*) \ldot \exists V \in \powerset^*(\Sigma^*) \ldot U \pref V \implies V \in \good(\lang(\varphi)) \lor V \in \bad(\lang(\varphi))$.
	As the set of good prefixes $\good(\lang(\varphi))$ is empty by Lemma~\ref{thm:forall-good-empty} we can simplify the formula to $\forall U \in \powerset^*(\Sigma^*) \ldot \exists V \in \powerset^*(\Sigma^*) \ldot U \pref V \implies V \in \bad(\lang(\varphi))$.
	Given an arbitrary $u \in \Sigma_\pathvars^*$, we translate it into the (canonical) $U'$ and get a $V'$ satisfying the conditions above.
	Let $v' \in \Sigma_\pathvars^*$ be the finite trace constructed from the extensions of $u$ in $V'$ (not canonical, but all are bad prefixes since $V' \in \bad(\lang(\varphi))$).
	By assumption, $u'v' \in \bad(\lang(\psi))$.
	\qed
\end{proof}

\begin{corollary}
	Given a $\hyperltl$ formula $\varphi = \exists \pi_1 \dots \exists \pi_k. \psi$, where $\psi$ is an $\ltl$ formula.
	$\varphi$ is monitorable if, and only if, $\forall u \in \Sigma_\pathvars^* \ldot \exists v \in \Sigma_\pathvars^* \ldot uv \in \good(\lang(\psi))$.
\end{corollary}

\begin{theorem}
	Given an alternation-free $\hyperltl$ formula $\varphi$.
	Deciding whether $\varphi$ is monitorable in the unbounded sequential model is $\pspace$-complete.
\end{theorem}
\begin{proof}
	We consider the case that $\varphi = \forall \pi_1 \dots \forall \pi_2 \ldot \psi$, the case for existentially quantified formulas is dual.
	We apply the characterization from Theorem~\ref{thm:monitorability-forall}.
	First, we have to check validity of $\psi$ which can be done in polynomial space~\cite{conf/stoc/SistlaC82}.
	Next, we have to determine whether $\forall u \in \Sigma_\pathvars^* \ldot \exists v \in \Sigma_\pathvars^* \ldot uv \in \bad(\lang(\psi))$.
	We use a slight modification of the $\pspace$ algorithm given by Bauer~\cite{journals/corr/abs-1006-3638}.
	Hardness follows as the problem is already $\pspace$-hard for $\ltl$.
	\qed
\end{proof}

\paragraph{Bounded Sequential Model.}
A Hyperproperty $H$ is \emph{monitorable} in the bounded input model if
\begin{equation*}
	\forall U \in \powerset^*(\Sigma^*) \ldot \exists V \in \powerset^*(\Sigma^*) \ldot U \pref V \wedge |U| = |V| \implies V \in \good(H) \lor V \in \bad(H) \enspace.
\end{equation*}

\begin{corollary}
	Given an arbitrarily $\hyperltl$ formula $\varphi$. Deciding whether $\varphi$ is monitorable in the bounded sequential model is $\pspace$-complete.
\end{corollary}

\paragraph{Parallel Model.}
A Hyperproperty $H$ is \emph{monitorable} in the fixed size input model if for a given bound $b$
\begin{equation*}
	\forall U \in \powerset^b(\Sigma^*) \ldot \exists V \in \powerset^b(\Sigma^*) \ldot U \pref V \implies V \in \good(H) \lor V \in \bad(H) \enspace.
\end{equation*}
This is a special case of the monitorability notion for the bounded input model, where the bound is fixed beforehand. Thus, the previous results are carried over to this model.
\begin{corollary}
	.Given an arbitrarily $\hyperltl$ formula $\varphi$. Deciding whether $\varphi$ is monitorable in the parallel model is $\pspace$-complete.
\end{corollary}
\section{Monitoring Hyperproperties in the Parallel Model}
\label{sec:fixed_size_monitoring}

We begin with the special case of the parallel model, where the set of traces is fixed in advance. Therefore, the traces can be processed in parallel, either online in a forward fashion, or offline in a backwards fashion. Before describing the respective algorithms, we define the finite trace semantics for $\hyperltl$ in the next subsection.

\subsection{Finite Trace Semantics}
We define a finite trace semantics for $\hyperltl$ based on the finite trace semantics of LTL~\cite{books/daglib/0080029}.
In the following, when using $\lang(\varphi)$ we refer to the finite trace semantics of a $\hyperltl$ formula $\varphi$.
Let $t$ be a finite trace, $\epsilon$ denotes the empty trace, and $|t|$ denotes the length of a trace. Since we are in a finite trace setting, $t[i,\ldots]$ denotes the subsequence from position $i$ to position $|t|-1$.
Let $\pathassignfin : \pathvars \rightarrow \Sigma^*$ be a partial function mapping trace variables to finite traces. We define $\epsilon[0]$ as the empty set.
$\pathassignfin[i, \ldots]$ denotes the trace assignment that is equal to $\pathassignfin(\pi)[i,\ldots]$ for all $\pi$. We define a subsequence of $t$ as follows.
\[
t[i,j] = \begin{cases}
\epsilon & \text{if } i \geq |t|\\
t[i,\textit{min}(j,|t|-1)], & \text{otherwise}
\end{cases}
\]
\begin{equation*}
\begin{array}{ll}
\pathassignfin \models_T a_\pi         \qquad \qquad & \text{if } a \in \pathassignfin(\pi)[0] \\
\pathassignfin \models_T \neg \varphi              & \text{if } \pathassignfin \nmodels_T \varphi \\
\pathassignfin \models_T \varphi \lor \psi         & \text{if } \pathassignfin \models_T \varphi \text{ or } \pathassignfin \models_T \psi \\
\pathassignfin \models_T \X \varphi                & \text{if } \pathassignfin[1,\ldots] \models_T \varphi \\
\pathassignfin \models_T \varphi\U\psi             & \text{if } \exists i \geq 0 \ldot \pathassignfin[i,\ldots] \models_T \psi \land \forall 0 \leq j < i \ldot \pathassignfin[j,\ldots] \models_T \varphi \\
\pathassignfin \models_T \exists \pi \ldot \varphi & \text{if there is some } t \in T \text{ such that } \pathassignfin[\pi \mapsto t] \models_T \varphi
\end{array}
\end{equation*}

\subsection{Monitoring Algorithm}

In this subsection, we describe our automata-based monitoring algorithm for the parallel model for $\hyperltl$.
We present an online algorithm that processes the traces in a forward fashion, and an offline algorithm that processes the traces backwards.

\paragraph{Online Algorithm.} For the online algorithm, we employ standard techniques for building LTL monitoring automata and use this to instantiate this monitor by the traces as specified by the $\hyperltl$ formula.
Let $\ap$ be a set of atomic propositions and $\pathvars = \{\pi_1, \ldots, \pi_n\}$ a set of trace variables. 
A deterministic monitor template $\monitor = (\Sigma, Q, \delta, q_0)$ is a four tuple of
a finite alphabet $\Sigma = 2^{\ap \times \pathvars}$,
a non-empty set of states $Q$,
a partial transition function $\delta: Q \times \Sigma \hookrightarrow Q$,
and a designated initial state $q_0 \in Q$.
The instantiated automaton runs in parallel over traces in $(2^\ap)^*$, thus we define a run with respect to a $n$-ary tuple $N \in ((2^\ap)^*)^n$ of finite traces.
A run of $N$ is a sequence of states $q_0 q_1 \cdots q_m \in Q^*$, where $m$ is the length of the smallest trace in $N$, starting in the initial state $q_0$ such that for all $i$ with $0 \leq i < m$ it holds that 
\begin{equation*}
\delta\left(q_i, \bigcup_{j=1}^n \bigcup_{a \in N(j)(i)} \set{(a,\pi_j)} \right) = q_{i+1} \enspace.
\end{equation*}
A tuple $N$ is accepted, if there is a run on $\monitor$.
For LTL, such a deterministic monitor can be constructed in doubly-exponential time in the size of the formula~\cite{conf/cav/dAmorimR05,journals/fmsd/TabakovRV12}.

\begin{example}
	We consider again the conference management example from the introduction.
	We distinguish two types of traces, \emph{author traces} and \emph{program committee member traces}, where the latter starts with proposition $pc$.
	Based on this traces, we want to verify that no paper submission is lost, i.e., that every submission (proposition $s$) is visible (proposition $v$) to every program committee member in the following step.
	When comparing two PC traces, we require that they agree on proposition $v$. The monitor template for the following $\hyperltl$ formalization is depicted in Fig.~\ref{fig:monitor-template}.
	\begin{equation} \label{eq:conference-management-visibility}
	\forall \pi. \forall \pi' \ldot \big((\neg pc_\pi \wedge pc_{\pi'}) \rightarrow \X\G (s_\pi \rightarrow \X v_{\pi'})\big) \wedge \big((pc_{\pi} \wedge pc_{\pi'}) \rightarrow \X\G (v_\pi \leftrightarrow v_\pi')\big)
	\end{equation}
\end{example}
\begin{figure}[t]
	\centering
	\begin{tikzpicture}[auto,->,>=stealth',shorten >=1pt,thick,transform shape, scale=0.8]
  \node[state,initial,initial text=] (init) {$q_0$};
  \node[state,below=of init] (accept) {$q_2$};
  \node[state,below left=1 and 2 of init] (wait-for-submit) {$q_1$};
  \node[state,below right=1 and 2 of init] (pc-equal) {$q_3$};
  \node[state,left=of wait-for-submit] (submitted) {$q_4$};
  
  \draw (init) edge node[swap] {$\neg pc_\pi \land pc_{\pi'}$} (wait-for-submit)
        (wait-for-submit) edge[loop right] node {$\neg s_p$} ()
        (wait-for-submit) edge[bend left] node {$s_p$} (submitted)
        (submitted) edge[bend left] node {$v_{\pi'}$} (wait-for-submit)
        (submitted) edge[loop left] node {$v_{\pi'} \land s_p$} ()
        (init) edge node {$\neg pc_{\pi'}$} (accept)
        (accept) edge[loop right] node {$\top$} ()
        (init) edge node {$pc_\pi \land pc_{\pi'}$} (pc-equal)
        (pc-equal) edge[loop right] node {$v_\pi \leftrightarrow v_{\pi'}$} ()
        ;
\end{tikzpicture}
	\caption{Visualization of a monitor template corresponding to formula given in Equation~\ref{eq:conference-management-visibility}. We use a symbolic representation of the transition function $\delta$.}
	\label{fig:monitor-template}
	\vspace{-10pt}
\end{figure}
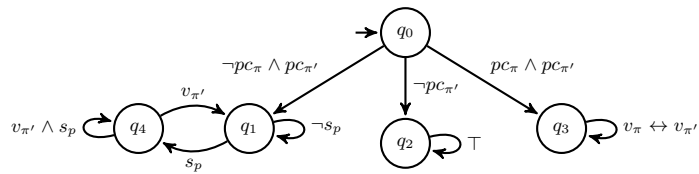


\begin{figure}
	\begin{algorithm}[H]
		\SetKwInOut{Input}{input}
		\SetKwInOut{Output}{output}
		\SetAlgoLined
		\Input{$Q^n$ HyperLTL formula $\varphi$}
		\Output{satisfied or $n$-ary tuple\\witnessing violation}
		\BlankLine
		$\monitor_\varphi = (\Sigma, Q, \delta, q_0) =$ \texttt{build\_template($\varphi$)}\;
		$S: T^n \rightarrow Q$\;
		$t_n \coloneqq (\epsilon_1,\ldots,\epsilon_n)$\;
		\BlankLine
		\While{$p^n \leftarrow$ new elements (in a forward fashion)}{
				$t \coloneqq (t_1~p_1,\ldots,t_n~p_n)$\;
				progress every state in $S$ according to $\delta$\;
				\If{$\lozenge_{1} \dots \lozenge_{n}~t_1 \times \dots \times t_n$ has violation in $M_\varphi$}{
					\Return witnessing tuple $t^n$\;
				}
			}
		\Return satisfied\;
		\BlankLine
	\end{algorithm}
\caption{Online algorithm for the parallel model, where $\lozenge_i := \wedge$ if $Q_i=\forall$ and $\vee$ otherwise.}
\label{alg:fixed_online}
\end{figure}

The \emph{online} algorithm is depicted in Fig.~\ref{alg:fixed_online}. It proceeds with the pace of the incoming stream.
We have a variable $S$ that maps tuples of traces to states of the deterministic monitor.
Whenever a trace progresses, we update the states in $S$ according to the transition function $\delta$.
If on this progress, there is a violation, we return the corresponding tuple of traces as a witness.

\paragraph{Offline Algorithm.}
For our efficient offline algorithm, we process the traces in a backwards fashion. The basic idea is to use an alternating automaton, which can be constructed in linear time, to find violations in the input set. The basic idea is to progress through the automaton in a bottom up fashion while proceeding backwards through the traces.

An \emph{alternating automaton}~\cite{DBLP:conf/banff/Vardi95}, whose runs generalize from sequences to trees, is a tuple $\mathcal A = (Q,q_0,\delta,\Sigma,F)$. $Q$ is the set of states, $q_0$ is the initial state, $\Sigma$ is the alphabet, and $F$ the set of accepting states. $\delta: Q \times \Sigma \rightarrow \mathbb{B}^+{Q}$ is a transition function, which maps a state and a symbol into a boolean combination of states. Thus, a run(-tree) of an alternating B\"uchi automaton $\mathcal A$ on an infinite word $w$ is a $Q$-labeled tree.
A word $w$ is accepted by $\mathcal{A}$ and called a \emph{model} if there exists a run-tree $T$ such that all paths $p$ through $T$ are accepting, i.e., $\textbf{Inf}(p) \cap F \not = \emptyset$. Note that an alternating automata can be constructed in linear time from an LTL formula~\cite{DBLP:conf/banff/Vardi95}.

\begin{figure}
	\begin{algorithm}[H]
		\SetKwInOut{Input}{input}
		\SetKwInOut{Output}{output}
		\SetAlgoLined
		\Input{HyperLTL formula $Q^n.\psi$\\trace set $T$}
		\Output{satisfied or $n$-ary tuple\\witnessing violation}
		\BlankLine
		$A_\psi = (\Sigma, Q, \delta, q_0) =$ \texttt{build\_alternating\_automaton($\psi$)}\;
		\texttt{booleans} $b^n$\;
		\BlankLine
			\ForEach{$t_i \in T$}{
			$b_i \coloneqq$\texttt{LTL\_backwards\_algorithm($A_\psi$,$t_i$)}
			}
			\If{$b_1 \lozenge_{1} \dots \lozenge_{n} b_n$}{\Return satisfied\;}
			\Return witnessing tuple $t^n$\;
		\BlankLine
	\end{algorithm}
\caption{Offline backwards algorithm for the parallel model, where $\lozenge_i := \wedge$ if $Q_i=\forall$ and $\vee$ otherwise.}
\label{alg:fixed_offline}
\end{figure}

The offline algorithm is depicted in Fig.~\ref{alg:fixed_offline}. The input is an arbitrary HyperLTL formula and a trace set $T$. After building the alternating automaton from the suffix of the HyperLTL formula, we apply the backwards monitoring algorithm~\cite{DBLP:journals/fmsd/FinkbeinerS04} for each trace $t_i \in T$. Depending on the quantifier prefix, we check the $n$-fold combinations of the traces for satisfaction. The difference to the forward algorithm is that we solely refer to previously computed values, resulting in a highly efficient algorithm.

\begin{theorem}
	Given a $\hyperltl$ formula $\varphi$ and a trace set $T$, the offline backwards algorithm for the parallel model runs in polynomial time in the size of the formula and the length of the traces.
\end{theorem}
\section{Monitoring Hyperproperties: Sequential}
\label{sec:sequential_monitoring}

The algorithms for monitoring $\hyperltl$ formulas in both sequential models (unbounded and bounded) are presented in Fig.~\ref{fig:algorithms_sequential}.
After building the deterministic monitoring automaton $\monitor_\varphi$, the algorithm proceeds with the pace of the incoming stream, which has an indicator when a new trace starts.
We have a variable $S$ that maps tuples of traces to states of the deterministic monitor.
Whenever a trace progresses, we update the states in $S$ according to the transition function $\delta$.
If during this process a violation is detected, we return the corresponding tuple of traces as a witness.
When a new trace $t$ starts, only new tuples are considered for $S$, that are tuples $N \in (T \cup \set{t})^n$ containing the new trace $t$, i.e., $N \notin T^n$.
In the bounded case, the algorithm implements a \emph{stop-condition}, namely when the trace set $T$ surpasses the given bound $b$. Note that monitoring formulas with alternations becomes possible, as their satisfaction or violation can be determined at least when the bound $b$ is reached.

\begin{figure}[t]
	\centering
			\begin{algorithm}[H]
				\SetKwInOut{Input}{input}
				\SetKwInOut{Output}{output}
				\SetAlgoLined
				\Input{$\forall^n$ HyperLTL formula $\varphi$ (unbounded case)\\HyperLTL formula $\varphi$ and bound $b$ (bounded case)}
				\Output{satisfied or $n$-ary tuple\\witnessing violation}
				\BlankLine
				$\monitor_\varphi = (\Sigma, Q, \delta, q_0) =$ \texttt{build\_template($\varphi$)}\;
				$S: T^n \rightarrow Q$\;
				$T \coloneqq \emptyset$\;
				$t \coloneqq \epsilon$\;
				\BlankLine
				\While{$p \leftarrow$ new element}{
					\eIf{$p$ is new trace (and $|T| \leq b$ in the bounded case)}{
						$T \cup \{t\}$\;
						$t \coloneqq \epsilon$\;
						$S \coloneqq \set{q_0 \mid \text{for new $n$-tuple}}$\;
					}{
						$t \coloneqq t~p$\;
						progress every state in $S$ according to $\delta$\;
						\If{violation}{
							\Return witnessing tuple\;
						}
					}
				}
				\Return satisfied\;
				\BlankLine
			\end{algorithm}
	\caption{Evaluation algorithm for monitoring $\forall^n$ HyperLTL formulas in the unbounded sequential model and monitoring arbitrary HyperLTL formulas in the bounded sequential model.}
	\label{fig:algorithms_sequential}
	\vspace{-10pt}
\end{figure}


In contrast to previous approaches, our algorithm returns a witness for violation.
This highly desired property comes with a price. In constructed worst case scenarios, we have to remember every system trace in order to return an explicit witness.
However, it turns out that practical hyperproperties satisfy certain properties such that the majority of traces can be pruned during the monitoring process.

\subsection{Optimizations}
\label{sec:minimizingtracestorage}
The main obstacle in monitoring hyperproperties in the unbounded input model is the potentially unbounded space consumption.
In the following, we present two analysis phases of our algorithm and describe how we incorporate tries at the base of our implementation. The first phase is a specification analysis, which is a preprocessing step that analyzes the HyperLTL formula under consideration. We use the recently introduced satisfiability solver for hyperproperties EAHyper~\cite{conf/cav/FinkbeinerHS17} to detect whether a formula is (1) \emph{symmetric}, i.e., we halve the number of instantiated monitors, (2) \emph{transitive}, i.e, we reduce the number of instantiated monitors to two, or (3) \emph{reflexive}, i.e., we can omit the self comparison of traces.
The second analysis phase is applied during runtime. We analyze the incoming trace to detect whether or not this trace poses strictly more requirements on future traces, with respect to a given HyperLTL formula.

\newpage

\subsubsection{Specification Analysis.}
\emph{Symmetry.}
Symmetry is particular interesting since many information flow policies satisfy this property. Consider, for example, observational determinism 
$
\mathit{ObsDet} \coloneqq \forall \pi\ldot \forall \pi'\ldot (O_\pi = O_{\pi'}) \W (I_\pi \neq I_{\pi'}).
$
We detect symmetry by translating this formula to a formula $\mathit{ObsDet}_\mathit{symm}$ that is unsatisfiable if there exists no set of traces for which every trace pair violates the symmetry condition:
\[
\mathit{ObsDet}_\mathit{symm} \coloneqq \exists \pi\ldot \exists \pi'\ldot \big((O_\pi = O_{\pi'}) \W (I_\pi \neq I_{\pi'})\big) \nleftrightarrow \big((O_\pi' = O_{\pi}) \W (I_\pi' \neq I_{\pi})\big)
\]
This is a sufficient condition for the invariance of $\mathit{ObsDet}$ under $\pi$ and $\pi'$, which we define in the following, and, therefore, $\mathit{ObsDet}$ is symmetric.
\begin{definition} \label{def:symmetry}
	Given a $\hyperltl$ formula $\varphi = \forall \pi_1 \dots \forall \pi_n \ldot \psi$, where $\psi$ is an $\ltl$ formula over trace variables $\{\pi_1, \dots, \pi_n\}$.
	We say $\varphi$ is invariant under trace variable permutation $\sigma : \pathvars \to \pathvars$, if for any set of traces $T \subseteq \Sigma^{\omega}$ and any assignment $\pathassign : \pathvars \to T$, $\pathassign \models_T \psi \Leftrightarrow (\pathassign \circ \sigma) \models_T \psi$.
	We say $\varphi$ is symmetric, if it is invariant under every trace variable permutation in $\pathvars \to \pathvars$.
\end{definition}

We generalize the previous example to formulas with more than two universal quantifiers. We use the fact, that the symmetric group for a finite set  $\pathvars$ of $n$ trace variables is generated by the two permutations $(\pi_1 \; \pi_2)$ and $(\pi_1 \; \pi_2 \; \cdots \; \pi_{n-1} \; \pi_n)$. If the $\hyperltl$-SAT solver determines that the input formula is invariant under these two permutations, then the formula is invariant under every trace variable permutation and thus symmetric.

\begin{theorem} \label{thm:symmetry_construction}
	Given a $\hyperltl$ formula $\varphi = \forall \pi_1 \dots \forall \pi_n \ldot \psi$, where $\psi$ is an $\ltl$ formula over trace variables $\{\pi_1, \dots, \pi_n\}$.
	$\varphi$ is symmetric if and only if
	$
	\varphi_\mathit{symm} = \exists \pi_1 \dots \exists \pi_n \ldot (\psi(\pi_1, \pi_2, \dots, \pi_{n-1}, \pi_n) \nleftrightarrow \psi(\pi_2, \pi_1, \dots, \pi_{n-1}, \pi_n))
	$
	$
	\vee (\psi(\pi_1, \pi_2, \dots, \pi_{n-1}, \pi_n) \nleftrightarrow \psi(\pi_2, \pi_3, \dots, \pi_n, \pi_1))
	$
	is unsatisfiable.
\end{theorem}

\paragraph{Transitivity.}
While symmetric HyperLTL formulas allow us to prune half of the monitor instances, transitivity of a HyperLTL formula has an even larger impact on the required memory. Observational Determinism, considered above, is not transitive.
However, equality, i.e, $\mathit{EQ} \coloneqq \forall \pi. \forall \pi' \ldot \G (a_\pi \leftrightarrow a_{\pi'})$, for example, is transitive and symmetric and allow us to reduce the number of monitor instances to one, since we can check equality against any reference trace.

\begin{definition} \label{def:transitivity}
	Given a $\hyperltl$ formula $\varphi = \forall \pi_1\ldot \forall \pi_2 \ldot \psi$, where $\psi$ is an $\ltl$ formula over trace variables $\{\pi_1, \pi_2\}$. Let $T = \{t_1, t_2, t_3\} \in \Sigma^{\omega}$ be three-elemented set of traces.
	We define the assignment $\pathassign_{i,j} : \pathvars \to \Sigma^{\omega}$ by $\pathassign_{i,j} \coloneqq \{\pi_1 \mapsto t_i, \pi_2 \mapsto t_j\}$.
	We say $\varphi$ is transitive, if $T$ was chosen arbitrary and $(\pathassign_{1,2} \models_{T} \psi) \wedge (\pathassign_{2,3} \models_{T} \psi) \Rightarrow \pathassign_{1,3} \models_{T} \psi$.
\end{definition}


\begin{theorem} \label{thm:transitivity_construction}
	Given a $\hyperltl$ formula $\varphi = \forall \pi_1 \ldot \forall \pi_2 \ldot \psi$, where $\psi$ is an $\ltl$ formula over trace variables $\{\pi_1, \pi_2\}$.
	$\varphi$ is transitive if and only if
	$
	\varphi_\mathit{trans} = \exists \pi_1 \exists \pi_2 \exists \pi_3 \ldot (\psi(\pi_1, \pi_2) \wedge \psi(\pi_2, \pi_3)) \nrightarrow \psi(\pi_1, \pi_3)
	$
	is unsatisfiable.
\end{theorem}

\paragraph{Reflexivity.}
Lastly, we introduce a method to check whether a formula is reflexive, which enables us to omit the composition of a trace with itself in the monitoring algorithm. Both $\hyperltl$ formulas considered in this section, $\mathit{ObsDet}$ and $\mathit{EQ}$, are reflexive.
\begin{definition} \label{def:reflexivity}
	Given a $\hyperltl$ formula $\varphi = \forall \pi_1 \dots \forall \pi_n \ldot \psi$, where $\psi$ is an $\ltl$ formula over trace variables $\{\pi_1, \dots, \pi_n\}$.
	We say $\varphi$ is reflexive, if for any trace $t \in \Sigma^{\omega}$ and the corresponding assignment $\pathassign : \pathvars \to \{t\}$, $\pathassign \models_{\{t\}} \psi$.
\end{definition}


\begin{theorem} \label{thm:reflexivity_construction}
	Given a $\hyperltl$ formula $\varphi = \forall \pi_1 \dots \forall \pi_n \ldot \psi$, where $\psi$ is an $\ltl$ formula over trace variables $\{\pi_1, \dots, \pi_n\}$.
	$\varphi$ is reflexive if and only if
	$
	\varphi_\mathit{refl} = \exists \pi \ldot \neg \psi(\pi, \pi, \dots, \pi)
	$
	is unsatisfiable.
\end{theorem}

\subsubsection{Trace Analysis.}
\label{trace-analysis}
In the previous subsection, we described a preprocessing step to reduce the number of monitor instantiations.
The main idea of the trace analysis, considered in the following, is to check whether a trace contains new requirements on the system under consideration.
If this is not the case, then this trace will not be stored by our monitoring algorithm.
We denote $\monitor_\varphi$ as the monitor template of a $\forall^*$ HyperLTL formula $\varphi$.

\begin{definition}
Given a HyperLTL formula $\varphi$, a trace set $T$ and an arbitrary $t \in \traces$, we say that $t$ is $(T,\varphi)$-redundant if $T$ is a model of $\varphi$ if and only if $T \cup \{t\}$ is a model of $\varphi$ as well. Formally denoted as follows.
\[
\forall T' \supseteq T.\; T' \in \lang(\varphi) \Leftrightarrow T' \cup \{t\} \in \lang(\varphi).
\]
\end{definition}

\begin{example}
	\label{example_conference}
	Consider, again, our example hyperproperty for a conference management system.
	\emph{``A user submission is immediately visible for every program committee member}
	\emph{and every program committee member observes the same.''}
We formalized this property as a $\forall^2$ HyperLTL formula in Equation~\ref{eq:conference-management-visibility}.
Assume our algorithm observes the following three traces of length five.
\begin{align}
&\tracebox{\hspace{0.5ex}\{\}}\tracebox{\set{s}}\tracebox{\hspace{0.5ex}\{\}} \tracebox{\hspace{0.5ex}\{\}} \tracebox{\hspace{0.5ex}\{\}} &&\textit{an author submits a paper} \label{one}\\
&\tracebox{\hspace{0.5ex}\{\}}\tracebox{\hspace{0.5ex}\{\}}\tracebox{\set{s}} \tracebox{\hspace{0.5ex}\{\}} \tracebox{\hspace{0.5ex}\{\}} &&\textit{an author submits a paper one time unit later}\label{two}\\
&\tracebox{\hspace{0.5ex}\{\}}\tracebox{\hspace{0.5ex}\{\}}\tracebox{\set{s}} \tracebox{\set{s}} \tracebox{\hspace{0.5ex}\{\}}  &&\textit{an author submits two papers}\label{three}
\end{align}
Trace~\ref{two} contains, with respect to $\varphi$ above, no more information than trace~\ref{three}. We say that trace~\ref{three} dominates trace~\ref{two} and, hence, trace~\ref{two} may be pruned from the set of traces that the algorithm has to store.
If we consider a PC member trace, we encounter the following situation.
\begin{align}
&\tracebox{\hspace{0.5ex}\{\}}\tracebox{\set{s}}\tracebox{\hspace{0.5ex}\{\}} \tracebox{\hspace{0.5ex}\{\}} \tracebox{\hspace{0.5ex}\{\}} &&\textit{an author submits a paper} \label{four}\\
&\tracebox{\hspace{0.5ex}\{\}}\tracebox{\hspace{0.5ex}\{\}}\tracebox{\set{s}} \tracebox{\set{s}} \tracebox{\hspace{0.5ex}\{\}} &&\textit{an author submits two papers}\label{five}\\
&\tracebox{\hspace{0.5ex}\{\}}\tracebox{\hspace{-0.5ex}\set{{\tiny{pc}}}}\tracebox{\set{v}}\tracebox{\set{v}} \tracebox{\set{v}} &&\textit{a PC member observes three submissions} \label{six}
\end{align}
Our algorithm will detect no violation, since the program committee member sees all three papers.
Intuitively, one might expect that no more traces can be pruned from this trace set.
However, in fact, trace~\ref{six} dominates trace~\ref{four} and trace~\ref{five}, since the information that three papers have been submitted is preserved in trace~\ref{six}.
Hence, it suffices to remember the last trace to detect, for example, the following violations.
\begin{align}
&\tracebox{\hspace{0.5ex}\{\}}\tracebox{\hspace{-0.5ex}\set{{\tiny{pc}}}}\tracebox{\set{v}} \tracebox{\set{v}} \tracebox{\set{v}} &&\textit{a PC member observes three submissions} \label{seven}\\
&\tracebox{\hspace{0.5ex}\{\}}\tracebox{\hspace{-0.5ex}\set{{\tiny{pc}}}}\tracebox{\set{v}} \tracebox{\set{v}}\tracebox{\hspace{0.5ex}\{\}}  &&\textit{\Lightning a PC member observes two submissions \Lightning} \label{eight}\\
&\textit{or} \nonumber\\
&\tracebox{\hspace{0.5ex}\{\}}\tracebox{\hspace{0.5ex}\{\}} \tracebox{\hspace{0.5ex}\{\}} \tracebox{\hspace{0.5ex}\{\}} \tracebox{\set{s}} &&\textit{\Lightning an author submits a non-visible paper \Lightning}
\end{align}
Note that none of the previous user traces, i.e., trace~\ref{one} to trace~\ref{five}, are needed to detect a violation.
\end{example}
\begin{definition}
	\label{def_dominance}
	Given $t,t'\in\traces$, we say $t$ \emph{dominates} $t'$ if $t'$ is $(\{t\},\varphi)$-redundant.
\end{definition}
The observations from Example~\ref{example_conference} can be generalized to a language inclusion check (cf. Theorem~\ref{dominating}), to determine whether a trace dominates another trace.
For proving this, we first prove the following two lemmas. For the sake of simplicity, we consider $\forall^2$ HyperLTL formulas. The proofs can be generalized. We denote $\monitor_{\varphi}[t/\pi]$ as the monitor where trace variable $\pi$ of the template Monitor $\monitor_{\varphi}$ is initialized with explicit trace $t$.
\begin{lemma} \label{thm:language-inclusion-by-iteration}
Let $\varphi$ be a $\forall^2$ HyperLTL formula over trace variables $\{\pi_1, \pi_2\}$.
	Given an arbitrary trace set $T$ and an arbitrary trace $t$, $T\cup \{t\}$ is a model of $\varphi$ if and only if $T$ is still accepted by the following two monitors: (1) only $\pi_1$ is initialized with $t$ (2) only $\pi_2$ is initialized with $t$. Formally, the following equivalence holds. 
	\[\forall T \subseteq \traces, \forall t \in \traces \ldot T \cup \{t\} \in \lang(\varphi) \Leftrightarrow T \subseteq \lang(\monitor_\varphi[t/\pi_1]) \wedge T \subseteq \lang(\monitor_\varphi[t/\pi_2])\]
\end{lemma}
\begin{lemma}
	Given a $\forall^2$ HyperLTL formula $\varphi$ over trace variables $\pathvars \coloneqq \{\pi_1, \ldots, \pi_n\}$ and two traces $t,t' \in \traces$, the following holds: $t$ dominates $t'$ if and only if
	\[
	\lang(\monitor_\varphi[t/\pi_1]) \subseteq \lang(\monitor_\varphi[t'/\pi_1]) \wedge \lang(\monitor_\varphi[t/\pi_2]) \subseteq \lang(\monitor_\varphi[t'/\pi_2])
	\]
\end{lemma}
\begin{proof}
  Assume for the sake of contradiction that (a) $t$ dominates $t'$ and w.l.o.g. (b) $\lang(\monitor_\varphi[t/\pi_1]) \nsubseteq \lang(\monitor_\varphi[t'/\pi_1])$.
  Thus, by definition of subset, there exists a trace $\tilde{t}$ with $\tilde{t} \in \lang(\monitor_\varphi[t/\pi_1])$ and $\tilde{t} \not \in \lang(\monitor_\varphi[t'/\pi_1])$.
  Hence, $\pathassign = \set{\pi_1 \mapsto t, \pi_2 \mapsto \tilde{t}}$ is a valid trace assignment, whereas $\pathassign' = \set{\pi_1 \mapsto t', \pi_2 \mapsto \tilde{t}}$ is not.
  On the other hand, from (a) the following holds by Definition~\ref{def_dominance}: $\forall T'$ with $\set{t} \subseteq T'$ it holds that $T' \in \lang(\varphi) \Leftrightarrow T' \cup \{t'\} \in \lang(\varphi)$.
  We choose $T'$ as $\{t,\tilde{t}\}$, which is a contradiction to the equivalence since we know from (a) that $\pathassign$ is a valid trace assignment, but $\pathassign'$ is not a valid trace assignment.

  For the other direction, assume that $\lang(\monitor_\varphi[t/\pi_1]) \subseteq \lang(\monitor_\varphi[t'/\pi_1])$ and $\lang(\monitor_\varphi[t/\pi_2]) \subseteq \lang(\monitor_\varphi[t'/\pi_2])$.
  Let $T'$ be arbitrary such that $\set{t} \subseteq T'$.
  We distinguish two cases:
  \begin{itemize}
  	\item
      Case $T' \in \lang(\varphi)$, then (a) $T' \subseteq \lang(M_\varphi[t/\pi_1]) \subseteq \lang(M_\varphi[t'/\pi_1])$ and (b) $T' \subseteq \lang(M_\varphi[t/\pi_2]) \subseteq \lang(M_\varphi[t'/\pi_2])$.
      By Lemma~\ref{thm:language-inclusion-by-iteration} and $T' \in \lang(\varphi)$, it follows that $T' \cup \set{t'} \in \lang(\varphi)$.
  \item
      Case $T' \notin \lang(\varphi)$, then $T' \cup \set{\hat{t}} \notin \lang(\varphi)$ for an arbitrary trace $\hat{t}$.
  \end{itemize}
\end{proof}
A generalization leads to the following theorem, which serves as the foundation of our trace storage minimization algorithm.
\begin{theorem}
  \label{dominating}
  Given a $\forall^n$ HyperLTL formula $\varphi$ over trace variables $\pathvars \coloneqq \{\pi_1, \ldots, \pi_n\}$ and two traces $t,t' \in \traces$, the following holds: $t$ dominates $t'$ if and only if
  \begin{equation*}
  \bigwedge_{\pi \in \mathcal{V}}\lang(\monitor_\varphi[t/\pi]) \subseteq \lang(\monitor_\varphi[t'/\pi]) \enspace.
  \end{equation*}
\end{theorem}

The characterization of dominance for existential quantification is dual.
\begin{lemma} \label{thm:exists-2-dominance}
	Given an $\exists^2$ HyperLTL formula $\varphi$ over trace variables $\pathvars := \{\pi_1, \ldots, \pi_n\}$ and two traces $t,t' \in \traces$, the following holds: $t$ dominates $t'$ if and only if
	\[
	\lang(\monitor_\varphi[t'/\pi_1]) \subseteq \lang(\monitor_\varphi[t/\pi_1]) \wedge \lang(\monitor_\varphi[t'/\pi_2]) \subseteq \lang(\monitor_\varphi[t/\pi_2])
	\]
\end{lemma}
\begin{proof}
  Assume for the sake of contradiction that (a) $t$ dominates $t'$ and w.l.o.g. (b) $\lang(\monitor_\varphi[t'/\pi_1]) \nsubseteq \lang(\monitor_\varphi[t/\pi_1])$.
  Thus, by definition of subset, there exists a trace $\tilde{t}$ with $\tilde{t} \in \lang(\monitor_\varphi[t'/\pi_1])$ and $\tilde{t} \not \in \lang(\monitor_\varphi[t/\pi_1])$.
  Hence, $\pathassign = \set{\pi_1 \mapsto t', \pi_2 \mapsto \tilde{t}}$ is a valid trace assignment, whereas $\pathassign' = \set{\pi_1 \mapsto t, \pi_2 \mapsto \tilde{t}}$ is not.
  On the other hand, from (a) the following holds by Definition~\ref{def_dominance}: $\forall T'$ with $\set{t} \subseteq T'$ it holds that $T' \in \lang(\varphi) \Leftrightarrow T' \cup \{t'\} \in \lang(\varphi)$.
  We choose $T'$ as $\{t,\tilde{t}\}$, which is a contradiction to the equivalence since we know from (a) that $\pathassign$ is a valid trace assignment, but $\pathassign'$ is not a valid trace assignment.

  For the other direction, assume that $\lang(\monitor_\varphi[t'/\pi_1]) \subseteq \lang(\monitor_\varphi[t/\pi_1])$ and $\lang(\monitor_\varphi[t'/\pi_2]) \subseteq \lang(\monitor_\varphi[t/\pi_2])$.
  Let $T'$ be arbitrary such that $\set{t} \subseteq T'$.
  We distinguish two cases:
  \begin{itemize}
  	\item
      Case $T' \cup \set{t'} \in \lang(\varphi)$, then (a) $T' \subseteq \lang(M_\varphi[t'/\pi_1]) \subseteq \lang(M_\varphi[t/\pi_1])$ and (b) $T' \subseteq \lang(M_\varphi[t'/\pi_2]) \subseteq \lang(M_\varphi[t/\pi_2])$.
      By Lemma~\ref{thm:language-inclusion-by-iteration} and $T' \cup \set{t'} \in \lang(\varphi)$, it follows that $T' \in \lang(\varphi)$.
  \item
      Case $T' \cup \set{t'} \notin \lang(\varphi)$, then $T' \notin \lang(\varphi)$.
  \end{itemize}
\end{proof}
\begin{corollary} \label{thm:exists-dominating}
  Given an $\exists^n$ HyperLTL formula $\varphi$ over trace variables $\pathvars \coloneqq \{\pi_1, \ldots, \pi_n\}$ and two traces $t,t' \in \traces$, the following holds: $t$ dominates $t'$ if and only if
  $
    \bigwedge_{\pi \in \mathcal{V}}\lang(\monitor_\varphi[t'/\pi]) \subseteq \lang(\monitor_\varphi[t/\pi])
  $.
\end{corollary}

\begin{figure}[t]
\centering
\scalebox{.9}{
\begin{algorithm}[H]
	\label{alg_minimization}
	\SetKwInOut{Input}{input}
	\SetKwInOut{Output}{output}
	\SetAlgoLined
	\Input{$\forall^n$ $\hyperltl$ formula $\varphi$,\\
		redundancy free set of traces $T$\\
		trace $t$}
	\Output{redundancy free set of traces $T_\mathit{min} \subseteq T \cup \set{t}$}
	\BlankLine
	$\monitor_\varphi =$ \texttt{build\_template($\varphi$)}
	\BlankLine
	\ForEach{$t' \in T$}{
		\If{$\bigwedge_{\pi \in \pathvars} \lang(\monitor_\varphi[t'/\pi]) \subseteq \lang(\monitor_\varphi[t/\pi])$}
		{
			return $T$
		}
	}
	\ForEach{$t' \in T$}{
		\If{$\bigwedge_{\pi \in \pathvars} \lang(\monitor_\varphi[t/\pi]) \subseteq \lang(\monitor_\varphi[t'/\pi])$}
		{
              $T \coloneqq T \setminus \set{t'}$
		}
	}
	\Return $T \cup \set{t}$
\end{algorithm}
}
\caption{Storage Minimization Algorithm.}
\end{figure}
\begin{theorem}\label{alg}
  Algorithm~\ref{alg_minimization} preserves the minimal trace set $T$, i.e., for all $t \in T$ it holds that $t$ is not $(T \setminus \set{t},\varphi)$-redundant.
\end{theorem}
\begin{proof}
	By induction on $T \setminus \set{t}$ and Theorem~\ref{dominating}.
\end{proof}

In the following, we give a characterization of the trace dominance for $\hyperltl$ formulas with one alternation.
These characterizations can be checked similarly to the algorithm depicted in Fig.~\ref{alg_minimization}.
\begin{theorem} \label{thm:forall-exists-dominance}
	Given a $\hyperltl$ formula $\forall \pi \ldot \exists \pi' \ldot \psi$ two traces $t,t' \in \traces$, the following holds: $t$ dominates $t'$ if and only if
	\begin{equation*}
	\lang(\monitor_\varphi[t/\pi]) \subseteq \lang(\monitor_\varphi[t'/\pi])
	\text{ and }
	\lang(\monitor_\varphi[t'/\pi']) \subseteq \lang(\monitor_\varphi[t/\pi']) \enspace.
	\end{equation*}
\end{theorem}
\begin{proof}
	The $\Rightarrow$ direction is the same as in proofs of Theorem~\ref{dominating} and Lemma~\ref{thm:exists-2-dominance}.
	
	For the other direction, assume that that (a) $\lang(\monitor_\varphi[t/\pi]) \subseteq \lang(\monitor_\varphi[t'/\pi])$ and (b) $\lang(\monitor_\varphi[t'/\pi']) \subseteq \lang(\monitor_\varphi[t/\pi'])$.
	Let $T'$ be arbitrary such that $\set{t} \subseteq T'$.
	We distinguish two cases:
	\begin{itemize}
		\item
		Case $T' \in \lang(\varphi)$, then for all $t_1 \in T'$ there is a $t_2 \in T'$ such that $\pathassignfin = \set{\pi \mapsto t_1, \pi' \mapsto t_2} \models_\emptyset \psi$.
		Especially, for $t$, there is a corresponding trace $t^*$ such that $\set{\pi \mapsto t, \pi' \mapsto t^*} \models_\emptyset \psi$, thus $t^* \in \lang(\monitor_\varphi[t/\pi])$.
		From (a) it follows that $t^* \in \lang(\monitor_\varphi[t'/\pi])$.
		Hence, $\set{\pi \mapsto t', \pi' \mapsto t^*} \models_\emptyset \psi$ and thereby $T' \cup \set{t'} \in \lang(\varphi)$.
		\item
		Case $T' \cup \set{t'} \in \lang(\varphi)$, then for all $t_1 \in T' \cup \set{t'}$ there is a $t_2 \in T' \cup \set{t'}$ such that $\set{\pi \mapsto t_1, \pi' \mapsto t_2} \models_\emptyset \psi$.
		Assume for the sake of contradiction there is a $t_1 \in T'$ such that there is no $t_2 \in T'$ with $\set{\pi \mapsto t_1, \pi' \mapsto t_2} \models_\emptyset \psi$.
		It follows that $\set{\pi \mapsto t_1, \pi' \mapsto t'} \models_\emptyset \psi$, i.e., $t_1 \in \lang(\monitor_\varphi[t'/\pi'])$.
		From (b) it follows that $t_1 \in \lang(\monitor_\varphi[t/\pi'])$, leading to the contradiction that $\set{\pi \mapsto t_1, \pi' \mapsto t} \models_\emptyset \psi$ and $t \in T'$.
		Hence, $T' \in \lang(\varphi)$.
	\end{itemize}
\end{proof}
\begin{corollary}
	Given a $\hyperltl$ formula $\exists \pi \ldot \forall \pi' \ldot \psi$ two traces $t,t' \in \traces$, the following holds: $t$ dominates $t'$ if and only if
	\begin{equation*}
	\lang(\monitor_\varphi[t'/\pi]) \subseteq \lang(\monitor_\varphi[t/\pi])
	\text{ and }
	\lang(\monitor_\varphi[t/\pi']) \subseteq \lang(\monitor_\varphi[t'/\pi']) \enspace.
	\end{equation*}
\end{corollary}

\begin{example}
	We show the effect of the dominance characterization on two example formulas.
	Consider the $\hyperltl$ formula $\forall \pi \ldot \exists \pi' \ldot \G (a_\pi \rightarrow b_{\pi'})$ and the traces $\set{b} \emptyset$, $\set{b}\set{b}$, $\set{a} \emptyset$, and $\set{a}\set{a}$.
	Trace $\set{a}\set{a}$ dominates trace $\set{a}\emptyset$ as instantiating $\pi$ requires two consecutive $b$'s for $\pi'$ where $\set{a}\emptyset$ only requires a $b$ at the first position (both traces do not contain $b$'s, so instantiating $\pi'$ leads to the same language).
	Similarly, one can verify that $\set{b}\set{b}$ dominates trace $\set{b}\emptyset$.
	
	Consider alternatively the formula $\exists \pi \ldot \forall \pi' \ldot \G (a_\pi \rightarrow b_{\pi'})$.
	In this case, $\set{a}\emptyset$ dominates $\set{a}\set{a}$ and $\set{b}\emptyset$ dominates $\set{b}\set{b}$.
\end{example}
For our conference management example formula given in Equation~\ref{eq:conference-management-alternating}, a trace $\set{pc}\emptyset\set{v}$ dominates $\set{pc}\emptyset\emptyset$ and $\emptyset\set{s}\emptyset$ dominates $\emptyset\emptyset\emptyset$, but $\emptyset\set{s}\emptyset$ and $\set{pc}\emptyset\set{v}$ are incomparable with respect to the dominance relation.

\subsubsection{Tries.}
After having specification and trace analysis in place, one observation is that there are a lot of incoming traces sharing same prefixes, leading to a lot of redundant monitor automaton instantiations, repetitive computations and duplicated information when those traces get stored.
Ideally one wants to avoid this overhead when monitoring a trace that shares some prefixes with already seen traces.
We implemented this idea with a so-called \emph{trie} datastructure.
Tries, also known as prefix trees, describe a tree data structure, which can represent a set of words over an alphabet in a compact manner, which exactly suites our use case.
The root of a trie is identified with the empty word $\epsilon$, additionally each node can have several child nodes, each of which corresponds to a unique letter getting appended to the representing word of the parent node. So the set of words of a trie is identified with the set of words the leaf nodes represent.

Formally,
    a trie is a four tuple $(\Sigma, \Tau, \longrightarrow, \tau_0)$ of
    a finite alphabet $\Sigma$,
    a non-empty set of states $\Tau$,
    a transition function $\longrightarrow: \Tau \times \Sigma \rightarrow \Tau$,
    and a designated initial state $\tau_0 \in \Tau$ called root of the trie.
    Instead of $((\tau,a),\tau') \in \longrightarrow$ we will write $\tau \overset{a}{\longrightarrow} \tau'$ in the following.
    For a trie to be of valid form we restrict $\longrightarrow$ such that, $\forall \tau,\tau' \in \Tau. |\{\tau \overset{a}{\longrightarrow} \tau'| a \in \Sigma\}| \leq 1$.


As mentioned before, storing the incoming traces succinctly is a major concern in monitoring hyperproperties.
In our case the alphabet would be the set of propositions used in the specification, and the word which is built from the trie are the traces.
Instead of storing each trace individually, we store all of them in one trie structure, branching only in case of deviation. This means equal prefixes only have to be stored once.
Besides the obvious benefits for memory, we also can make use of the maintained trie data structure to improve the runtime of our monitoring algorithms. As traces with same prefixes end up corresponding to the same path in the trie, we only have to instantiate the monitor automaton as much as the trie contains branches.
\begin{figure}
	\begin{algorithm}[H]
		\SetKwInOut{Input}{input}
		\SetKwInOut{Output}{output}
		\SetAlgoLined
		\Input{$\forall^n$ HyperLTL formula $\varphi$}
		\Output{satisfied or $n$-ary tuple witnessing violation}
		\BlankLine
		$\monitor_\varphi = (\Sigma, Q, \delta, q_0) =$ \texttt{build\_template($\varphi$)}\;
		$S: \Tau^n \rightarrow Q$\; $\tau_0 \coloneqq $\texttt{new\_trie()}\; $\mathbf{i} \coloneqq \tau_0^n$\; $I \coloneqq \{\mathbf{i}\}$\;
		\BlankLine
		\While{$\mathbf{p} \leftarrow$ new elements (in a forward fashion)}{
			\For{$1 \leq j \leq n$}{
				$\mathbf{i}(j) \leftarrow $\texttt{add\_value($\mathbf{i}(j)$, $j$, $\mathbf{p}(j)$)}\;
			}
			$I \leftarrow \bigcup_{\mathbf{i} \in I} \{(i'_1,\ldots, i'_n)| \mathbf{i}(j) \overset{a}{\longrightarrow} i'_j, a \in \Sigma, 1 \leq j \leq n\}$\;
			\ForEach{$\mathbf{i} \in I$}{
				progress every state in $S$ according to $\delta$\;
				\If{violation in $M_\varphi$}{
					$t \coloneqq ($\texttt{rooted\_sequence($i(1)$)}$,\ldots,$\texttt{rooted\_sequence($i(n)$)}$)$\;
					\Return witnessing tuple $t^n$\;
				}
			}
		}
		\Return satisfied\;
		\BlankLine
	\end{algorithm}
	\caption{Online algorithm using trie datastructure.}
	\label{alg:trie_online_parallel}
\end{figure}

\paragraph{Algorithm.}
Our trie-based parallel monitoring algorithm is depicted in Fig.~\ref{alg:trie_online_parallel}.
Without using tries our monitoring algorithm was based on instantiating the deterministic monitor template $\monitor_\phi$ with tuples of traces. Now we instantiate $\monitor_\phi$ with tuples of tries. Initially we only have to create the single instance having the the root of our trie.
The trie-based algorithm has much in common with its previously discused trace-based pendant. Initially we have to build the determintistic monitor automaton $\monitor_\varphi = (\Sigma, Q, \delta, q_0)$. We instantiate the monitor with a fresh trie root $\tau_0$. A mapping from trie instantiations to a state in $\monitor_\varphi$ $S: \Tau^n \rightarrow Q$, stores the current state of our   For each of the incoming traces we provide an entry in a tuple of tries $\mathbf{\tau}$, each entry gets initialized to $\tau_0$. During the run of our algorithm these entries will get updated such that they always correspond to the word built by the traces up to this point.
For as long as there are traces left, which have not yet ended, and we have not yet detected a violation, we will proceed updating the entries in $\mathbf{\tau}$ as follows. Having entry $\tau$ and the correspond trace sequence proceeds with $a$, if $\exists \tau' \in \Tau. \tau \overset{a}{\longrightarrow} \tau'$, we update the entry to $\tau'$ otherwise we create such a child node of $\tau$. Creating a new node in the trie always occures, if the prefix of the incoming trace starts to differ from already seen prefixes.
After having moved on step in our traces sequences, we have to reflect this step in our trie structure, in order for the trie-instantiated automata to correctly monitor the new propositions. As a trie node can branch to multiple child nodes, each monitor instantiation will get replaced by the set of instantiations, where all possible child combinations of the different assigned tries are existant. $S$ will get updated such that those new tuples are mapped to the same state as the instantiation they were build from. We essential fork the monitor instantiation for the different branches in the trie.
After this preprocessing we are able to update our mapping $S$ according to $\delta$. If a violation is detected here, we will return the corresponding counter example as a tuple of traces, as those can get reconstructed by stepping upwards in the tries of $\tau$.

Note that the sequential algorithm is derived from the parallel one by progressing one trace at a time and additionally when forking the monitor instantiations, only those new instantiation have to be kept, which contain the (single) trie in $\mathbf{\tau}$ at least once. 

\section{Evaluation}
\label{sec:experimentalresults}

In the following section, we report on experimental results of our tool RVHyper v2. We briefly describe implementation details before evaluating our monitoring algorithms for both, the parallel and unbounded input model.

\subsection{Implementation}

We implemented the monitoring algorithm for the sequential input model in a tool called RVHyper\footnote{The implementation of RVHyper is available at \url{https://react.uni-saarland.de/tools/rvhyper/}.}~\cite{rvhyper}.
We extended this implementation to RVHyper v2, including the trie optimization technique and added an implementation of the online monitoring approach for the parallel input model.
RVHyper v2 is written in C\nolinebreak[4]\hspace{-.05em}\raisebox{.4ex}{\relsize{-3}{\textbf{++}}}.
It uses \emph{spot} for building the deterministic monitor automata and the \emph{Buddy} BDD library for handling symbolic constraints.
We use the $\hyperltl$ satisfiability solver EAHyper~\cite{conf/cav/FinkbeinerHS17,conf/concur/FinkbeinerH16} to determine whether the input formula is reflexive, symmetric, or transitive.
Depending on those results, we omit redundant tuples in the monitoring algorithm.


\subsection{Experimental Results: Sequential Input Model}
In this subsection, we report on experimental results of the algorithm for the sequential input model and, especially, the accompanying optimizations.
\paragraph{Specification Analysis.}
For the specification analysis, we checked  variations of observational determinism, quantitative non-interference~\cite{conf/cav/FinkbeinerRS15}, equality and our conference management example for symmetry, transitivity, and reflexivity. The results are depicted in Table~\ref{formulaanalysis}. The specification analysis comes with low costs (every check was done in under a second), but with a high reward in terms of constructed monitor instances (see Fig.~\ref{fig:runtime-comparison}).
For hyperproperties that do not satisfy one of the properties, e.g., our conference management example, our trace analysis will still dramatically reduce the memory consumption.
\begin{table}[t]
	\centering
	\caption{Specification Analysis for universally quantified hyperproperties.}
	\label{formulaanalysis}
	\begin{tabular}{ll|l|l|l|}
		\cline{3-5}
		&                                                                                                                                                                                                                                                                                                                                                                                                                                                                    & symm & trans & refl \\ \hline
		\multicolumn{1}{|l|}{ObsDet1}     & $\forall \pi. \forall \pi'.\; \G (I_{\pi} = I_{\pi'}) \rightarrow \G (O_{\pi} = O_{\pi'})$                                                                                                                                                                                                                                                                                                                                                             &  \cmark    &    \xmark   &  \cmark    \\ \hline
		\multicolumn{1}{|l|}{ObsDet2}     & $\forall \pi. \forall \pi'.\; (I_{\pi} = I_{\pi'}) \rightarrow \G (O_{\pi} = O_{\pi'})$                                                                                                                                                                                                                                                                                                                                                                &  \cmark    &    \xmark   &  \cmark      \\ \hline
		\multicolumn{1}{|l|}{ObsDet3}     & $\forall \pi. \forall \pi'. (O_\pi = O_\pi') \W (I_\pi \neq I_\pi')$                                                                                                                                                                                                                                                                                                                                                                                                   &  \cmark    &    \xmark   &  \cmark      \\ \hline
		\multicolumn{1}{|l|}{QuantNoninf} & $\forall \pi_0 \ldots \forall \pi_{c}.~\neg ((\bigwedge_i I_{\pi_i} = I_{\pi_0}) \wedge \bigwedge_{i \neq j} O_{\pi_i} \neq O_{\pi_j})$                                                                                                                                                                                                                                                                                                                            &  \cmark    &    \xmark   &  \cmark      \\ \hline

		\multicolumn{1}{|l|}{EQ} &\begin{tabular}[c]{@{}l@{}}$\forall \pi. \forall \pi' \ldot \G (a_\pi \leftrightarrow a_{\pi'})$\end{tabular} &  \cmark    &    \cmark   &  \cmark \\ \hline
		
		\multicolumn{1}{|l|}{ConfMan} &\begin{tabular}[c]{@{}l@{}}$\forall \pi \forall \pi' \ldot \big((\neg pc_\pi \wedge pc_{\pi'}) \rightarrow \X\G (s_\pi \rightarrow \X v_{\pi'})\big)$\\ $\wedge \big((pc_{\pi} \wedge pc_{\pi'}) \rightarrow \X\G (v_\pi \leftrightarrow v_\pi')\big)$\end{tabular} &  \xmark    &    \xmark   &  \xmark \\ \hline
	\end{tabular}
\end{table}

\paragraph{Trace Analysis.} For evaluating our trace analysis, we use a scalable, bounded variation of observational determinism:
$\forall \pi\ldot \forall \pi'\ldot \G_{<n} (I_\pi = I_{\pi'}) \rightarrow \G_{< n +c} (O_\pi = O_{\pi'})$. Figure~\ref{fig:number-of-pruned-traces} shows a family of plots for this benchmark class, where $c$ is fixed to three. We randomly generated a set of $10^5$ traces. The blue (dashed) line depicts the number of traces that need to be stored, the red (dotted) line the number of traces that violated the property, and the green (solid) line depicts the pruned traces. When \emph{increasing the requirements} on the system, i.e., decreasing $n$, we prune the majority of incoming traces with our trace analysis techniques.

\begin{figure}[t]
	\centering
	\begin{minipage}{.33\textwidth}
		\begin{tikzpicture}
		\begin{semilogyaxis}[tiny,width=1.2\textwidth,mark size=1.3pt,ymin=0,ymax=100000,xmin=0,xmax=100000,no markers,thick,xtick={0,25000,50000,75000,100000},xlabel={$n=16$}]
		\addplot+[blue,dashed] table {plots/12_16__3/nli1_ir16_nlo1_or3_nhi0_nho0_s.dat};
		\addplot+[green,solid] table {plots/12_16__3/nli1_ir16_nlo1_or3_nhi0_nho0_p.dat};
		\addplot+[red,dotted] table {plots/12_16__3/nli1_ir16_nlo1_or3_nhi0_nho0_v.dat};
		
		\end{semilogyaxis}
		\end{tikzpicture}
	\end{minipage}\begin{minipage}{.33\textwidth}
	\begin{tikzpicture}
	\begin{semilogyaxis}[tiny,width=1.2\textwidth,mark size=1.3pt,ymin=0,ymax=100000,xmin=0,xmax=100000,no markers,thick,xtick={0,25000,50000,75000,100000},xlabel={$n=14$}]
	\addplot+[blue,dashed] table {plots/12_16__3/nli1_ir14_nlo1_or3_nhi0_nho0_s.dat};
	\addplot+[green,solid] table {plots/12_16__3/nli1_ir14_nlo1_or3_nhi0_nho0_p.dat};
	\addplot+[red,dotted] table {plots/12_16__3/nli1_ir14_nlo1_or3_nhi0_nho0_v.dat};
	
	\end{semilogyaxis}
	\end{tikzpicture}
    \end{minipage}\begin{minipage}{.33\textwidth}
		\begin{tikzpicture}
		\begin{semilogyaxis}[tiny,width=1.2\textwidth,mark size=1.3pt,ymin=0,ymax=100000,xmin=0,xmax=100000,no markers,thick,xtick={0,25000,50000,75000,100000},xlabel={$n=12$}]
		\addplot+[blue,dashed] table {plots/12_16__3/nli1_ir12_nlo1_or3_nhi0_nho0_s.dat};
		\addplot+[green,solid] table {plots/12_16__3/nli1_ir12_nlo1_or3_nhi0_nho0_p.dat};
		\addplot+[red,dotted] table {plots/12_16__3/nli1_ir12_nlo1_or3_nhi0_nho0_v.dat};
		
		\end{semilogyaxis}
		\end{tikzpicture}
	\end{minipage}
	\caption{Absolute numbers of violations in red (dotted), number of instances stored in blue (dashed), number of instances pruned in green (solid) for $10^5$ randomly generated traces of length $100000$. The $y$ axis is scaled logarithmically.}
	\label{fig:number-of-pruned-traces}
\end{figure}
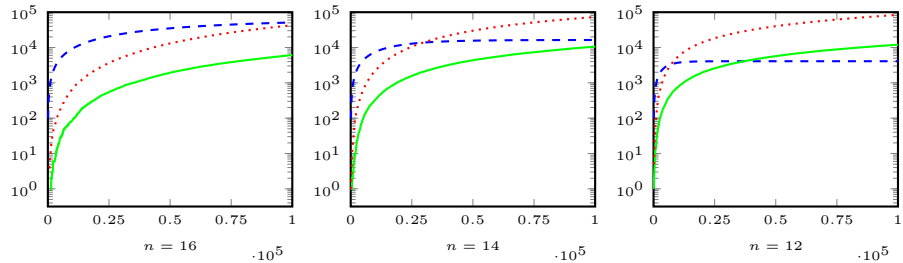

\paragraph{Optimizations in Combination.}
We furthermore considered our optimizations in different combination. Note that the trace analysis and the trie optimization are not trivially combinable and we leave this as future work.
As a first benchmark, we monitored an encoder for its robustness against errors, i.e., a Hamming-distance preserving encoder. That an encoder preserves a certain Hamming-distance can be encoded as a universally quantified HyperLTL formula~\cite{conf/post/ClarksonFKMRS14}.
In Fig.~\ref{fig:runtime-comparison} we compare the running time of the monitoring optimizations presented in this paper to the naive approach.
We compare the naive monitoring approach to the monitor using specification analysis and trace analysis, as well as a combination thereof.
We randomly built traces of length $50$. In each position of the trace, the corresponding bit had a 1\% chance to be flipped.
Applying our techniques results in a tremendous speed up of the monitoring algorithm, where the combination of specification analysis and trie optimization turns out to be superior for this benchmark.
\begin{figure}[t]
	\centering
	\begin{tikzpicture}
		\begin{axis}[width=0.55\textwidth,mark size=1.3pt,ymin=0,ymax=56296,xmin=0,xmax=500,no markers,thick,xlabel={\# of instances},ylabel={runtime in msec.},
            legend entries={naive,specification analysis,trace analysis,both,trie naive,trie},
  		    legend style={
              at={(-0.2,1)},
              anchor=north east}]
            ]]
            \addplot+[red,solid] table {plots/runtime/journal/encoder1_2_l50/timings_naive.dat};
            \addplot+[blue,dashed,very thick] table {plots/runtime/journal/encoder1_2_l50/timings_sa.dat};
            \addplot+[orange,dotted,very thick] table {plots/runtime/journal/encoder1_2_l50/timings_ta.dat};
            \addplot+[green,dashdotted,very thick] table {plots/runtime/journal/encoder1_2_l50/timings_both.dat};
            \addplot+[yellow,dashdotted,very thick] table {plots/runtime/journal/encoder1_2_l50/timings_trie_naive.dat};
            \addplot+[violet,dotted,very thick] table {plots/runtime/journal/encoder1_2_l50/timings_trie.dat};
		\end{axis}
	\end{tikzpicture}
	\caption{Hamming-distance preserving encoder: runtime comparison of naive monitoring approach with different optimizations and a combination thereof.}
	\label{fig:runtime-comparison}
\end{figure}
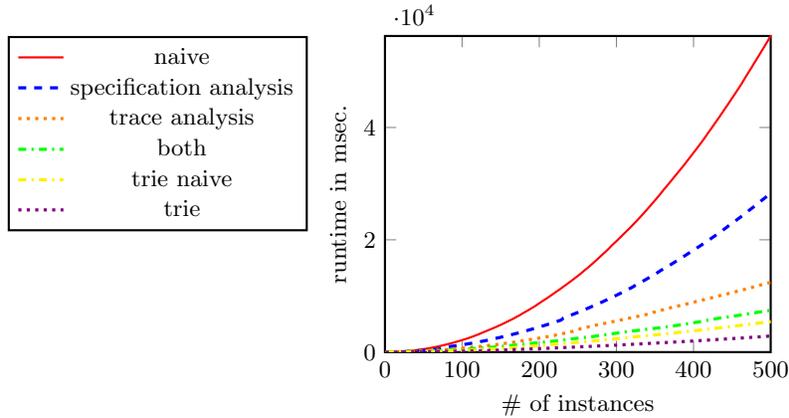

For our second benchmark, we considered a black box combinatorial circuit, guarded by a multiplexer that selects between the two input vectors $\vec i$ and $\vec i'$ and an inverse multiplexer that forwards the output of the black box either towards $\vec o$ or $\vec o'$. We monitored whether there is a semantic dependency between the in- and outputs (see~\cite{rvhyper} for details).
With random simulation, we generated traces of length $30$. Fig~\ref{fig:runtime-comparison-prob} demonstrates that the specification analysis, if applicable, is a valuable addition to both storage optimizations.
Note that both axes are log-scaled.
\begin{figure}[t]
	\centering
	\begin{tikzpicture}
        \begin{axis}[width=0.55\textwidth,mark size=1.3pt,ymode=log,ymax=400000,xmode=log,xmin=1,xmax=5000,no markers,thick,xlabel={probability for bit flip $\times 10^{-4}$},ylabel={runtime in msec.},
            legend entries={naive,specification analysis,trace analysis,both,trie naive,trie},
  		    legend style={
              at={(-0.2,1)},
              anchor=north east}]
            ]]
            \addplot+[red,solid] table {plots/runtime/journal/mux_30_probs/naive.dat};
            \addplot+[blue,dashed,very thick] table {plots/runtime/journal/mux_30_probs/sa.dat};
            \addplot+[orange,dotted,very thick] table {plots/runtime/journal/mux_30_probs/ta.dat};
            \addplot+[green,dashdotted,very thick] table {plots/runtime/journal/mux_30_probs/both.dat};
            \addplot+[yellow,dashdotted,very thick] table {plots/runtime/journal/mux_30_probs/trie_naive.dat};
            \addplot+[violet,dotted,very thick] table {plots/runtime/journal/mux_30_probs/trie.dat};
		\end{axis}
	\end{tikzpicture}
	\caption{Monitoring of black box circuits: runtime comparison of naive monitoring approach with different optimizations and a combination thereof.}
	\label{fig:runtime-comparison-prob}
\end{figure}
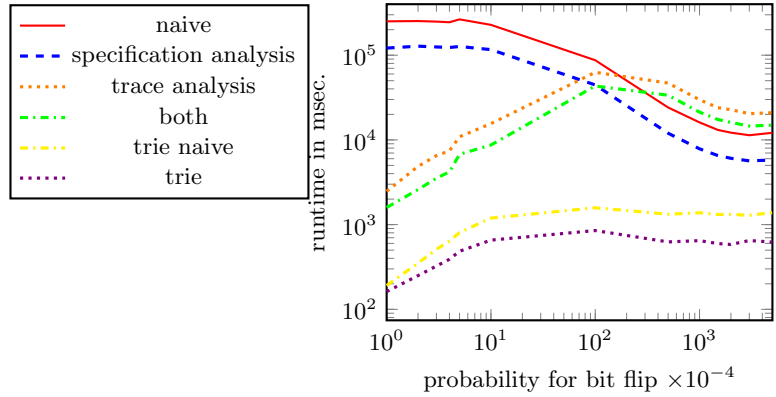

\paragraph{Use Case: Detecting Spurious Dependencies in Hardware Designs.}

The problem whether input signals influence output signals in hardware designs, was considered in~\cite{rvhyper}. We briefly describe the input specification and the corresponding hardware designs, before reporting the results of RVHyper v2 on those benchmarks.
We write $\vec i \ninfluences \vec o$ to denote that the inputs $\vec i$ do not influence the outputs $\vec o$.
Formally, we specify this property as the following $\hyperltl$ formula:
\begin{equation*}
\forall \pi_1 \forall \pi_2 \ldot
(\vec{o}_{\pi_1} = \vec{o}_{\pi_2}) \WUntil (\overline{\vec{i}}_{\pi_1} \neq \overline{\vec{i}}_{\pi_2}) \enspace,
\end{equation*}
where $\overline{\vec{i}}$ denotes all inputs except $\vec i$.
Intuitively, the formula asserts that for every two pairs of execution traces $(\pi_1,\pi_2)$ the value of $\vec{o}$ has to be the same until there is a difference between $\pi_1$ and $\pi_2$ in the input vector $\overline{\vec{i}}$, i.e., the inputs on which $\vec o$ may depend.

We applied RVHyper in both versions to traces generated from the following hardware designs (see~\cite{rvhyper} for details).

\begin{example}[\textsc{xor}]
	As a first example, consider the \textsc{xor} function $\vec{o} = \vec{i} \oplus \vec{i}'$.
	In the corresponding circuit, every $j$-th output bit $o_j$ is only influenced by the $j$-the input bits $i_j$ and $i'_j$.
\end{example}

\begin{example}[\textsc{mux}]
	We consider a black box combinatorial circuit, guarded by a multiplexer that selects between the two input vectors $\vec i$ and $\vec i'$ and an inverse multiplexer that forwards the output of the black box either towards $\vec o$ or $\vec o'$.
	Despite there being a syntactic dependency between $\vec o$ and $\vec i'$, there is no semantic dependency, i.e., the output $\vec o$ does solely depend on $\vec i$ and the selector signal.
	
	When using the same example, but with a sequential circuit as black box, there may be information flow from the input vector $\vec i'$ to the output vector $\vec o$ because the state of the latches may depend on it.
	We construct such a circuit that leaks information about $\vec{i}'$ via its internal state.
	
\end{example}

\begin{example}[counter]
	\label{ex:counter}
	Our last example is a binary counter with two input control bits $\mathit{incr}$ and $\mathit{decr}$ that increments and decrements the counter.
	The counter has a single output, namely a signal that is set to one when the counter value overflows.
	Both inputs influence the output, but timing of the overflow depends on the number of counter bits.
\end{example}

\begin{table}[]
	\caption{Results of RVHyper v2 compared to RVHyper v1 on traces generated from circuit instances. Every instance was run 10 times with different seeds and the average is reported.}
	\label{tbl:rvhyper-results}
	\centering
	\begin{tabular}{lllllll}
		\hline \noalign{\smallskip}
		instance  & \#\,traces  & \#\,instances v1 & \#\,instances v2 & time v1 & time v2 \\ \noalign{\smallskip}\hline\noalign{\smallskip}
		\textsc{xor1}  & 18 & 222 & 18 & 12ms & 6ms \\
		\textsc{xor2}  & 1000 & 499500 & 127 & 16913ms & 1613ms   \\
		count1  & 1636 & 1659446 & 2 & 28677ms & 370ms \\
		count2 & 1142 &  887902 & 22341 & 15574ms & 253ms\\
		\textsc{mux}  & 1000 &  499500 & 32 & 14885ms & 496ms \\
		\textsc{mux2}  & 82 & 3704 & 1913 & 140ms & 27ms \\ \hline
	\end{tabular}
\end{table}

\subsection{Experimental Results: Parallel Input Model}
In this subsection, we present experimental results of our optimizations for the online algorithm in the parallel input model.
As a first benchmark, we reused the Hamming-distance preserving encoder example described above. We generated $500$ traces of length $30$. The results are depicted in Figure~\ref{fig:runtime-comparison-par-prob}. The x-axis denotes the probability that a bit of the trace is flipped. The most left hand side of the scale means that 1 out of 10000 bits is flipped and the most right hand side of the scale means that every second bit is flipped. As expected the trie optimization has a major influence on the runtime of the monitoring algorithm as long as the traces potentially share the same prefix. The specification analysis, however, enhances the monitoring process regardless of a possible prefix equality.

\begin{figure}[t]
	\centering
	\begin{tikzpicture}
        \begin{axis}[width=0.55\textwidth,mark size=1.3pt,ymode=log,ymax=100000,xmode=log,xmin=1,xmax=5000,no markers,thick,xlabel={probability for bit flip $\times 10^{-4}$},ylabel={runtime in msec.},
	legend entries={naive,specification analysis,trie naive,trie},
	legend style={
		at={(-0.2,1)},
		anchor=north east}]
	]]
	\addplot+[red,solid] table {plots/runtime/journal/encoder2_3_l30_n500_probs_parallel/naive.dat};
	\addplot+[blue,dashed,very thick] table {plots/runtime/journal/encoder2_3_l30_n500_probs_parallel/sa.dat};
	\addplot+[yellow,dashdotted,very thick] table {plots/runtime/journal/encoder2_3_l30_n500_probs_parallel/trie_naive.dat};
	\addplot+[violet,dotted,very thick] table {plots/runtime/journal/encoder2_3_l30_n500_probs_parallel/trie.dat};
	\end{axis}
	\end{tikzpicture}
    \caption{Hamming-distance preserving encoder: runtime comparison of parallel monitoring approach with different optimizations.}
	\label{fig:runtime-comparison-par-prob}
\end{figure}
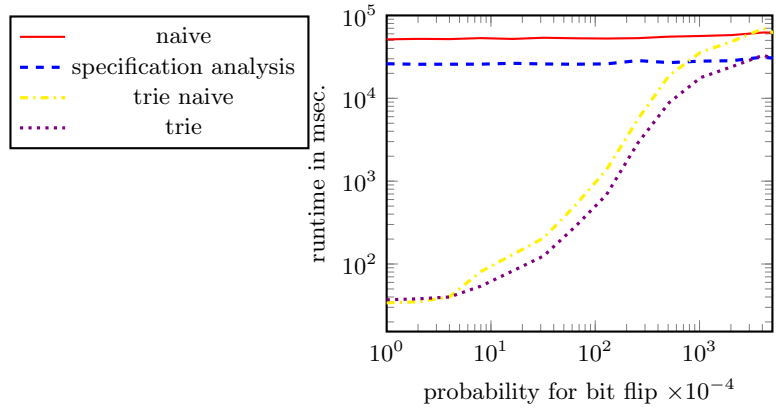


\section{Conclusion}
\label{sec:conclusion}
We have presented automata-based monitoring algorithms for $\hyperltl$.
We considered three different input models, where the traces are either given in parallel or sequentially, and, when given sequentially, may either grow beyond any bound or be limited by a fixed bound. 
We showed that deciding whether a $\hyperltl$ formula is monitorable in the parallel and bounded sequential models is $\pspace$-complete.
We showed that deciding whether an alternation-free formula is monitorable in the unbounded sequential model is $\pspace$-complete.

We presented three optimizations tackling different problems in monitoring hyperproperties.
The trace analysis minimizes the needed memory, by minimizing the stored set of traces.
The specification analysis reduces the algorithmic workload by reducing the number of comparisons between a newly observed trace and the previously stored traces.
The succinct representation of the trace set as a trie tackles the massive storage and computation overhead for prefix-equal traces.

We have evaluated our tool implementation RVHyper on several benchmarks, showing that the optimizations contribute significantly towards the practical monitoring of hyperproperties.



\bibliographystyle{plain}
\bibliography{main}
\end{document}